%% file: main-tops.tex
\documentclass[conference]{IEEEtran}
\ifCLASSINFOpdf
  \usepackage[pdftex]{graphicx}
\else
\fi
\usepackage[cmex10]{amsmath}
\usepackage{amsthm,amsfonts,amssymb}
\usepackage{ifthen}
\usepackage{tikz}
\usepackage{cite}
\usepackage{amsthm,amsfonts,amssymb}
\usepackage{graphicx}
\usepackage{subfig}
\usepackage{blkarray}
\usepackage{dsfont}
\usepackage[mathscr]{euscript}
\usepackage{enumitem}
\usepackage{amsfonts}
\usepackage{mathrsfs}

\usepackage{algorithm}
\usepackage[noend]{algpseudocode}
\usepackage{blkarray}
\usepackage{stfloats}%
\newtheorem{theorem}{Theorem}[section]
\newtheorem{corollary}{Corollary}[section]
\newtheorem{proposition}{Proposition}[section]
\newtheorem{lemma}{Lemma}[section]
\theoremstyle{remark}
\newtheorem*{remark}{Remark}
\theoremstyle{definition}
\newtheorem{definition}{Definition}[section]

\usepackage{url}
\usepackage{tcolorbox}
\usepackage{arydshln}
\usepackage{mathtools}
\usepackage{dsfont}
\definecolor{brickred}{cmyk}{0,0.89,0.94,0.28}
\definecolor{goldenrod}{cmyk}{0,0.10,0.84,0}
\definecolor{purple}{cmyk}{0.45,0.86,0,0}
\definecolor{rawsienna}{cmyk}{0,0.72,1,0.45}
\definecolor{olivegreen}{cmyk}{0.64,0,0.95,0.40}
\definecolor{peach}{cmyk}{0,0.5,0.7,0}
\definecolor{darkolive}{rgb}{0.,0.4,0.}
\colorlet{grey}{gray!40}

\global\long\def\E{\mathbb{E}}

\global\long\def\d{\mathrm{d}}

\usepackage{xcolor}

\usepackage{xcolor}

\usepackage[cmintegrals]{newtxmath}

\begin{document}
%
\title{Mean Estimation with User-Level Privacy for Spatio-Temporal IoT Datasets}



\author{\IEEEauthorblockN{
V. Arvind Rameshwar,
Anshoo Tandon,
Prajjwal Gupta,
Aditya Vikram Singh,\\
Novoneel Chakraborty,
and
Abhay Sharma
}
\thanks{V.~A.~R., A.~T., N.~C., and A.~S. are with the India Urban Data Exchange (IUDX) Program Unit, Indian Institute of Science, Bengaluru, India, emails: \texttt{arvind.rameshwar@gmail.com, anshoo.tandon@gmail.com,  novoneel.chakraborty@datakaveri.org, abhay.sharma@datakaveri.org}. P.~G. contributed to this work during an internship at IUDX; he is currently with Northeastern University, Boston, MA 02115, email: \texttt{gupta.praj@northeastern.edu}. A.~V.~S. is a PhD student at the Dept. of ECE, IISc, Bengaluru, email: \texttt{adityavs@iisc.ac.in}.}
}
 \IEEEoverridecommandlockouts


\maketitle

\begin{abstract}

This paper considers the problem of the private release of sample means of speed values from traffic datasets. Our key contribution is the development of  user-level differentially private algorithms that incorporate carefully chosen parameter values to ensure low estimation errors on real-world datasets, while ensuring privacy. We test our algorithms on ITMS (Intelligent Traffic Management System) data from an Indian city, where the speeds of different buses are drawn in a potentially non-i.i.d. manner from an unknown distribution, and where the number of speed samples contributed by different buses is potentially different. We then apply our algorithms to large synthetic datasets, generated based on the ITMS data. Here, we provide theoretical justification for the observed performance trends, and also provide recommendations for the choices of algorithm subroutines that result in low estimation errors. Finally, we characterize the best performance of pseudo-user creation-based algorithms on worst-case datasets via a minimax approach; this then gives rise to a novel procedure for the creation of pseudo-users, which optimizes the worst-case total estimation error. The algorithms discussed in the paper are readily applicable to general spatio-temporal IoT datasets for releasing a differentially private mean
of a desired value.


\end{abstract}


%
\IEEEpeerreviewmaketitle

\section{Introduction}

It is now well-understood that the release of even seemingly innocuous functions of a dataset that is not publicly available can result in the reconstruction of the identities of individuals (or users) in the dataset with alarming levels of accuracy (see, e.g., \cite{sweeney,narayanan}). A notable such reconstruction attack involved a somewhat na\"ively anonymized database of taxi data, released by the Taxi and Limousine Commission of New York City \cite{taxi}, which was succesfully deanonymized \cite{pandurangan}, thereby revealing sensitive information about the taxi drivers. To alleviate concerns over such attacks, the framework of differential privacy (DP) was introduced in \cite{dwork06}, which, informally speaking, guarantees the privacy of a \emph{single} data sample, or equivalently, of users when each user contributes at most one sample. However, most real-world datasets, such as traffic databases, record multiple contributions from every user; a straightforward application of standard DP techniques achieves poor estimation errors, owing to the addition of a large amount of noise to guarantee privacy. Recent work on ``user-level privacy'' \cite{userlevel} however demonstrates the effectiveness of some new algorithms that guarantee much improved estimation error due to the additional privacy requirement for (a fixed) $m>1$ samples per user.

In this paper, we provide algorithms, which draw on the research in \cite{userlevel}, for ensuring user-level privacy in the context of releasing the sample means of speed records in traffic datasets. Clearly, it is desirable to keep the speed values of vehicles private, because they indirectly reflect the individual driving behaviour and might affect vehicle insurance premiums. Our algorithms for estimating the sample means of the data crucially rely on carefully chosen procedures that first create \emph{pseudo-users}, or arrays, following \cite{tyagi}, and then clip the number of speed samples contributed by each user and clip each speed sample to lie in a high-probability interval. These procedures are designed with the objective of controlling the ``user-level sensitivity'' of the sample mean that we are interested in. 

We first emperically evaluate the performance of such algorithms (via their estimation errors) on real-world speed values from ITMS (Intelligent Traffic Management System) traffic data, supplied by IoT devices deployed in an Indian city. Here, the speeds of different buses are drawn in a potentially non-i.i.d. manner from an unknown distribution, and the number of speed samples contributed by different buses is potentially different. Next, we artificially generate a ``large'' synthetic dataset, using the statistics of the real-world ITMS data, with either a large number of users or a large number of samples contributed per user. We demonstrate, via extensive experiments, the effectiveness or the relative poor performance of the different algorithms we employ, in each case. In addition, we provide theoretical justification for the performance trends that we observe and recommendations for the choice of algorithm to be used on large real-world datasets. We mention that the results presented in this paper can be directly applied to Floating Car Data (FCD) (see, e.g., \cite{floating_car}) for estimating the traffic conditions in a particular location at a given time of the day.

Next, we consider the general setting of user-level differentially private mean estimation (with Laplace noise addition) using pseudo-user creation, where only the number of samples contributed by a user is clipped. For this setting, we present a ``minimax'' analysis of the total estimation error due to clipping and due to noise addition for privacy, and discuss some interesting consequences. In particular, as a by-product, we obtain an upper bound on the total error incurred by the ``best'' pseudo-user creation-based algorithm on \emph{any} dataset. This analysis then naturally gives rise to a novel procedure, based on the creation of pseudo-users, which clips the number of samples contributed by a user in such a manner as to jointly optimize the \emph{worst-case errors} due to clipping and noise addition. However, since the subroutine used in such a procedure involves a potentially large numerical optimization problem, we present and analytically solve a simpler \emph{convex} optimization problem that seeks to minimize a suitable upper bound on the worst-case error. 

The paper is organized as follows: Section \ref{sec:prelim} presents the problem formulation and relevant information on the ITMS dataset, and recapitulates preliminaries on DP and user-level DP; Section \ref{sec:alg} contains description of our main algorithms; Section \ref{sec:results} discusses the results obtained by running our algorithms on real-world ITMS and synthetic data, and provides theoretical proofs of the performance trends and recommendations for the choices of algorithms on large datasets. Section \ref{sec:minimax} presents a minimax analysis of the total estimation error incurred by a pseudo-user creation-based algorithm, and Section \ref{sec:opt-array-av} describes the novel algorithm for minimizing the worst-case total estimation error that results from the minimax analysis. The paper is then concluded in Section \ref{sec:conclusion}.



\section{Preliminaries}
\input{prelim.tex}
\section{Algorithms}
\input{alg.tex}

\section{Results}
\input{results.tex}

\section{On Minimax Error Bounds for Pseudo-User Creation-Based Algorithms}
\input{minimax.tex}
\section{The \textsc{OPT-Array-Averaging} Algorithm}
\input{opt-array-av.tex}
\section{Conclusion}
\input{conclusion.tex}

\section*{Acknowledgment}
The authors thank Prof. Himanshu Tyagi for helpful discussions.




%
\bibliographystyle{IEEEtran}
{\footnotesize
	\bibliography{references}}

\end{document}

%% file: prelim.tex
\label{sec:prelim}
\subsection{Notation}
For a given $n\in \mathbb{N}$, the notation $[n]$ denotes the set $\{1,2,\ldots,n\}$. Given a collection of real numbers $(x_1,\ldots,x_n)$, we use the notation med$(x_1,\ldots,x_n)$ to refer to a median of the collection, i.e., med$(x_1,\ldots,x_n)$ is any value $x$ such that $|\{i\in [n]: x_i\geq x\}| = \left \lceil n/2\right \rceil$. We write $X\sim P$ to denote that the random variable $X$ is drawn from the distribution $P$. We use the notation $\text{Lap}(b)$ to refer to the zero-mean Laplace distribution with standard deviation $\sqrt{2}b$; its probability distribution function (p.d.f.) obeys
\[
f(x) = \frac{1}{2b}e^{-|x|/b}, \ x\in \mathbb{R}.
\]
We also use the notation $\mathcal{N}(\mu,\sigma^2)$ to denote the Gaussian distribution with mean $\mu$ and variance $\sigma^2$. For a random variable $X\sim \mathcal{N}(0,1)$, we denote its complementary cumulative distribution function (c.c.d.f.) by $Q$, i.e., for $x\in \mathbb{R}$,
\[
Q(x):=\Pr[X\geq x] = \int_{x}^{\infty} \frac{1}{\sqrt{2\pi}} e^{-z^2/2}\d z.
\]
\subsection{Problem Setup}
The ITMS dataset that we consider contains records of the data provided by IoT sensors deployed in an Indian city, pertaining to information on bus movement. Each record catalogues, among other information, the license plate of the bus, the location at which the data was recorded, a timestamp, and the actual data value itself, which is the speed of the bus. For the purpose of analysis, we divide the total area of the city into hexagon-shaped grids, with the aid of Uber's Hexagonal Hierarchical Spatial Indexing System (or H3) \cite{h3}, which provides an open-source library for such partitioning tasks. Furthermore, we quantize the timestamps present in the data records into 1 hour timeslots. In this work, we consider those data records that pertain to a single \underline{H}exagon \underline{A}nd \underline{T}imeslot (or HAT), and seek to release the average speeds of the buses in the chosen HAT, privately, to a client who has no prior knowledge of these values. We remark that the algorithms discussed in this paper are readily applicable to general spatio-temporal IoT datasets for releasing a differentially private mean of a desired value. 

\subsection{Problem Formulation}
Fix a HAT of interest. Let $L$ be the number of users (or distinct license plates) present in the HAT, and for every user $\ell\in [L]$, let the number of records contributed by the user be $m_\ell$. We set $m^\star:= \max_{\ell\in [L]} m_\ell$ and $m_\star:= \min_{\ell\in [L]} m_\ell$ as the maximum and minimum number of records contributed by any user, respectively. We assume that $L$ and the collection $\{m_\ell: \ell \in [L]\}$ are known to the client. Now, let the collection $\left\{S_j^{(\ell)}: \ell\in [L], j\in [m_\ell]\right\}$ denote the speed values present in the records corresponding to the chosen HAT. We assume that each $S_j^{(\ell)}$ is a real number that lies in the interval $[0,U]$, where $U$ is a fixed upper bound on the speeds of the buses; $U = 65$ km/hr for the datasets we consider. For the real-world ITMS dataset that we work with, the speed samples are drawn according to some unknown distribution $P$ that is potentially non-i.i.d. (independent and identically distributed) across samples and users. However, when we generate synthetic speed samples, we draw each speed value $S_j^{(\ell)}$ i.i.d. according to some distribution $P_s$ that is obtained by analyzing the statistics of the ITMS data.

Call the database consisting of the speed records of users present in the chosen HAT as $\mathcal{D} = \left\{\left(u_\ell,S_j^{(\ell)}\right): \ell \in [L], j\in [m_\ell]\right\}$, where the collection $\{u_\ell: \ell\in [L]\}$ denotes the set of users. 

The function that we are interested in computing, which is also called a ``query'' to the database, is the sample average
\begin{equation}
	\label{eq:f}
f(\mathcal{D}):= \frac{1}{\sum_{\ell=1}^L m_\ell}\cdot \sum_{\ell=1}^L \sum_{j=1}^{m_\ell} S_j^{(\ell)}.
\end{equation}

Next, we define the notion of privacy that we are interested in, that is user-level differential privacy.

\subsection{User-Level Differential Privacy}
Consider two databases $\mathcal{D}_1 = \left\{\left(u_{\ell},x_j^{(\ell)}\right): \ell \in [L], j\in [m_\ell]\right\}$ and $\mathcal{D}_2 = \left\{\left(u_\ell,\overline{x}_j^{(\ell)}\right): \ell \in [L], j\in [m_\ell]\right\}$ consisting of the same users, with each user contributing the same number of (potentially different) data values $\{x_j^{(\ell)}\}$. Let $\mathsf{D}$ denote a universal set of such databases. We say that $\mathcal{D}_1$ and $\mathcal{D}_2$ are ``user-level neighbours'' if there exists $\ell_0\in [L]$ such that $\left(x^{(\ell_0)}_{1},\ldots, x^{(\ell_0)}_{m_{\ell_0}}\right)\neq \left(\overline{x}^{(\ell_0)}_{1},\ldots, \overline{x}^{(\ell_0)}_{m_{\ell_0}}\right)$, with $\left(x^{(\ell)}_{1},\ldots, x^{(\ell)}_{m_{\ell}}\right)= \left(\overline{x}^{(\ell)}_{1},\ldots, \overline{x}^{(\ell)}_{m_{\ell}}\right)$, for all $\ell\neq \ell_0$. In this work, we concentrate on mechanisms that map a given database to a single real value.

\begin{definition}
	For a fixed $\varepsilon>0$, a mechanism $M: \mathsf{D}\to \mathbb{R}$ is said to be user-level $\varepsilon$-DP if for every pair of datasets $\mathcal{D}_1, \mathcal{D}_2$ that are user-level neighbours, and for every measurable subset $Y \subseteq \mathbb{R}$, we have that
	\[
	\Pr[M(\mathcal{D}_1) \in Y] \leq e^\varepsilon \Pr[M(\mathcal{D}_2) \in Y].
	\]
\end{definition}
Next, we define the user-level sensitivity of a function of interest.
\begin{definition}
	Given a function $g: \mathsf{D}\to \mathbb{R}$, we define its user-level sensitivity $\Delta_g$ as
	\[
	\Delta_g:= \max_{\mathcal{D}_1,\mathcal{D}_2\ \text{u-l nbrs.}} \left|g(\mathcal{D}_1) - g(\mathcal{D}_2)\right|,
	\]
	where the maximization is over datasets that are user-level neighbours.
\end{definition}
For example, the user-level sensitivity of $f$, in \eqref{eq:f}, is
\begin{equation}\Delta_f = \frac{Um^\star}{\sum_\ell m_\ell}, \label{eq:delf}\end{equation} where $m^\star = \max_\ell m_\ell$. In this paper, we use the terms ``sensitivity'' and ``user-level sensitivity'' interchangeably.

The next result is well-known and follows from standard DP results \cite[Prop. 1]{dwork06}:

\begin{theorem}
	\label{thm:dp}
	For any $g: \mathsf{D}\to \mathbb{R}$, the mechanism $M^{\text{Lap}}_g: \mathsf{D}\to \mathbb{R}$ defined by
	\[
	M^{\text{Lap}}_g(\mathcal{D}_1) = g(\mathcal{D}_1)+Z,
	\]
	where $Z\sim \text{Lap}(\Delta_g/\varepsilon)$ is user-level $\varepsilon$-DP.
\end{theorem}

All the algorithms presented in this paper involve the addition of Laplace noise to guarantee user-level $\varepsilon$-DP, with the intuition that under some regularity conditions, for small $\varepsilon$ (or equivalently, high privacy requirements), Laplace distributed noise is asymptotically optimal in terms of the magnitude of error in estimation (see \cite{stair2,stair1}).

Furthermore, by standard results on the tail probabilities of Laplace random variables, we obtain the following bound on the estimation error:
\begin{proposition}
	For a given function $g:\mathsf{D}\to \mathbb{R}$ and for any dataset $\mathcal{D}_1$, we have that
	\[
	\Pr\left[\left|M^{\text{Lap}}_g(\mathcal{D}_1) - g(\mathcal{D}_1)\right|\geq  \frac{\Delta_g \ln(1/\delta)}{\varepsilon}\right]\leq \delta,
	\]
	for all $\delta\in [0,1]$.
\end{proposition}
%

Observe from Theorem \ref{thm:dp} that for a fixed $\varepsilon>0$, the standard deviation of the noise added is proportional to the product $Um^\star$; in settings where either $m^\star$ or $U$ is large, a prohibitively large amount of noise is added for privacy, thereby increasing the error in estimation.



%
%


%% file: alg.tex
\label{sec:alg}

In this section, we present 4 algorithms: \textsc{Baseline}, \textsc{Array-Averaging}, \textsc{Levy}, and \textsc{Quantile}, for private mean estimation with user-level privacy. These algorithms are designed keeping in mind their intended application, which is to analyze speed values of buses in large cities. Before we do so, we describe a useful approach that modifies the function $f$ to be estimated (see \eqref{eq:f}), in order to reduce the dependence of the standard deviation of the noise added on the number of samples per user (see the discussion at the end of the previous section): we clip the number of samples contributed by any user $\ell\in [L]$ to $\min\{m_\ell,m_\text{UB}\}$, where $m_\star\leq m_\text{UB}\leq m^\star$ could depend on $\{m_\ell\}_{\ell\geq 1}$. We describe next two such so-called ``grouping" strategies, which leads to the creation of ``pseudo-users'' or arrays, inspired by the work in \cite{tyagi}: \textsc{WrapAround} and \textsc{BestFit}.
\subsection{Strategies for Grouping Samples}
\label{sec:grouping}
Let $m_{\text{UB}}$ be given. The procedure underlying grouping strategies is as follows: for ease of exposition, we first sort the users in non-increasing order of the number of samples contributed, i.e., we reindex the users so that $m_1\geq m_2\geq\ldots\geq m_L$. We then initialize $L$ empty arrays $A_1,\ldots,A_L$, each of length $m_{\text{UB}}$. We assume that the locations of the arrays are indexed from 1 to $m_\text{UB}$, and for a given array $A$, we use the notation $w(A)$ to denote the number of filled locations. We initialilze $w(A_i) = 0$, for all $1\leq i\leq L$. Now, we process each user in turn, beginning with user 1 and populate his/her samples in the arrays, with a maximum of $m_\text{UB}$ samples from any user being populated in the arrays. The exact strategy for populating arrays followed in the \textsc{WrapAround} and \textsc{BestFit} procedures is explained below.
\subsubsection{\textsc{WrapAround}}
Let $\ell^*$ denote the smallest value of $\ell\in [L]$ such that $m_\ell < m_\text{UB}$. For every $\ell<\ell^*$, we populate the array $A_\ell$ with the speed samples $\left(S_j^{(\ell)}: 1\leq j\leq m_\text{UB}\right)$, in a contiguous manner. Next, we pick the samples $\left(S_j^{(\ell^\star)}: 1\leq j\leq m_{\ell^\star}\right)$ from user $\ell^\star$ and add these to array $A_{\ell^\star}$ contiguously. By the definition of $\ell^\star$, there exist some empty slots in array $A_{\ell^\star}$. These are filled completely by the next user's samples $\left(S_j^{(\ell^\star+1)}: 1\leq j\leq m_{\ell^\star+1}\right)$, if $m_{\ell^\star+1}\leq m_\text{UB}-m_{\ell^\star}$ contiguously; else, these slots are filled by the samples $\left(S_j^{(\ell^\star+1)}: 1\leq j\leq m_\text{UB}-m_{\ell^\star}\right)$, and the remaining samples from user $\ell^\star+1$ are filled in the first $m_{\ell^\star+1}+m_{\ell^\star} - m_{\text{UB}}$ positions of array $A_{\ell^\star+1}$. The procedure then continues, filling each array in turn, until all the user's samples have been populated in the arrays. Let $K$ denote the index of the last filled-up array, with the arrays $A_{K+1},\ldots,A_L$ being deleted. Note that
\begin{equation}
	\label{eq:K}
K = \left \lfloor \frac{\sum_{\ell=1}^{L}\min \left \{m_{\ell}, m_\text{UB} \right \}}{m_\text{UB}} \right \rfloor.
\end{equation}
Observe that in this strategy, it can happen that the samples from a given user $\ell\in [L]$ are split between two arrays $A_i$ and $A_{i+1}$, for some $1\leq i\leq K-1$; we then say that this user $\ell$ ``influences'' two arrays. The \textsc{WrapAround} strategy is shown as Algorithm \ref{alg:wrap}.


\subsubsection{\textsc{BestFit}}

This algorithm that we use is a version of a popular online algorithm for the bin-packing problem, which is a known strongly NP-complete problem (see \cite{np}). Similar to the description of \textsc{WrapAround}, let $\ell^*$ denote the smallest value of $\ell\in [L]$ such that $m_\ell < m_\text{UB}$. For every $\ell<\ell^*$, we populate the array $A_\ell$ with the speed samples $\left(S_j^{(\ell)}: 1\leq j\leq m_\text{UB}\right)$, in a contiguous manner. For each user $\ell\geq \ell^\star$, we find the least-indexed array $A$ among $A_{\ell^\star},\ldots,A_L$ that can accommodate $m_\ell$ samples and is filled the most. The $m_\ell$ samples from the user $\ell$ are then placed in array $A$ (in a continguous manner), and the process is iterated over the other users. Note that the number of non-empty arrays that are returned by the \textsc{BestFit} algorithm $\overline{K}$, is at least $K$ (see \eqref{eq:K}). Furthermore, observe that the samples from any given user $\ell\in [L]$ are present in exactly one array, unlike the case in the \textsc{WrapAround} strategy. Hence, any user influences at most one array, in this case.

The \textsc{BestFit} strategy is shown as Algorithm \ref{alg:bestfit}.



\begin{algorithm}
	\caption{The wrap-around strategy}
	\begin{algorithmic}[1]
		\Procedure {\textsc{WrapAround}}{$\mathcal{D}, m_\text{UB}$}
		
		\State Set $K \gets \left \lfloor \frac{\sum_{\ell=1}^{L}\min \left \{m_{\ell}, m_\text{UB} \right \}}{m_\text{UB}} \right \rfloor$.
		\State Initialize $t\gets 1$
		\For {$\ell= 1:L$}
		\State Set $r \gets \min\left \{m_{\ell}, m_\text{UB} \right \}$, $w \gets w(A_t)$.
		\If {$m_\text{UB}-w\geq r$}
		\State Fill $A_t(w+1:w+r) = \left(S_1^{(\ell)},\ldots, S_r^{(\ell)}\right)$
		\Else
		\State Fill $A_t(w+1:m_\text{UB}) = \left(S_1^{(\ell)},\ldots, S_{m_\text{UB}-w}^{(\ell)}\right)$.
		\State Update $t\gets t+1$, set $r' = r+w-m_\text{UB}$.
		\State Fill $A_t(1:r') = \left(S_{m_\text{UB}-w+1}^{(\ell)},\ldots, S_r^{(\ell)}\right)$.
		\EndIf
		\EndFor
		\State \Return $A_{1}, \ldots, A_{K}$
		\EndProcedure
	\end{algorithmic}
	\label{alg:wrap}
\end{algorithm}

\begin{algorithm}
	\caption{The best-fit strategy}
	\begin{algorithmic}[1]
		\Procedure{BestFit}{$\mathcal{D}, m_\text{UB}$}
		\For {$\ell \leftarrow 1:L$}
		\State Set $r \gets \min\left \{m_{\ell}, m_\text{UB}\right\}$
		\State Set $\mathcal{A} = \{A_i: m_\text{UB}-w(A_i)\geq r\}$.
		\State Set $A = A_j\in \mathcal{A}$ such that $j$ is the least index such that $w(A_j) = \max_{A'\in \mathcal{A}} w(A')$.
		\State Fill $A(w(A)+1:r+w(A)) = (S_1^{(\ell)},\ldots, S_r^{(\ell)})$.
		\EndFor
		\State \Return non-empty arrays $A_{1}, \ldots,A_{\overline{K}}$
		\EndProcedure
	\end{algorithmic}
\label{alg:bestfit}
\end{algorithm}

With these grouping strategies in place, we shall now describe our main algorithms.

\subsection{\textsc{Baseline}}
\label{sec:baseline}
The \textsc{Baseline} algorithm simply adds the required amount of Laplace noise to guarantee $\varepsilon$-DP, to the function $f$ to be computed, i.e.,
\[
M_\text{Baseline}(\mathcal{D}) = f(\mathcal{D})+\text{Lap}(\Delta_f/\varepsilon),
\]
where $\Delta_f = \frac{Um^\star}{\sum_\ell m_\ell}$ is the sensitivity of the function $f$. As mentioned earlier, this algorithm suffers from the drawback that a large amount of noise needs to be added for privacy when either $U$ or $m^\star$ is large, thereby impeding estimation. All the algorithms presented next attempt to ameliorate this by the creation of pseudo-users, or arrays.

\subsection{\textsc{Array-Averaging}}

In this algorithm, we attempt to reduce the amount of noise added when $m^\star$ is large, by clipping the number of samples contributed by any given user to some $m_\star\leq m_\text{UB}\leq m^\star$. We use one of the two grouping strategies discussed previously: \textsc{WrapAround} or \textsc{BestFit} for clipping. In particular, given the database $\mathcal{D}$, we set

\[
f_\text{arr, wrap}(\mathcal{D}):=\frac{1}{K}\cdot \sum_{i=1}^K \overline{A_i},
\]
where $A_i$ are the arrays obtained using the \textsc{WrapAround} strategy, with $\overline{A_i}:= \frac{1}{w(A_i)} \sum_{j=1}^{w(A_i)} A_i(j)$ being the mean of the samples contributed by array $i\in [K]$. We also set
\[
 f_\text{arr, best}(\mathcal{D}):=\frac{1}{\overline{K}}\cdot \sum_{i=1}^{\overline{K}} \overline{A_i},
\]
 where $A_i$ are the arrays obtained using the \textsc{BestFit} strategy, with $\overline{A_i}:= \frac{1}{w(A_i)} \sum_{j=1}^{w(A_i)} A_i(j)$ being the mean of the samples contributed by array $i\in [\overline{K}]$.
 
 From the discussion in Section \ref{sec:grouping}, we have that since using the \textsc{WrapAround} strategy can result in one user influencing two arrays, the sensitivity is given by
 \begin{equation*}
 \Delta_{f_\text{arr, wrap}} = \frac{2U}{K}.
 \label{eq:delfwrap}
 \end{equation*}
However, since using the \textsc{BestFit} strategy, any user influences at most array, the sensitivity is 
 \begin{equation}
	\Delta_{f_\text{arr, best}} = \frac{U}{\overline{K}}.
	\label{eq:delfbest}
\end{equation}

The \textsc{Array-Averaging} algorithm with \textsc{WrapAround} grouping computes
\[
M_\text{ArrayAvg, wrap}(\mathcal{D}) = f_\text{arr, wrap}(\mathcal{D})+\text{Lap}(\Delta_{f_\text{arr, wrap}}/\varepsilon),
\]
and the \textsc{Array-Averaging} algorithm with \textsc{BestFit} grouping computes
\[
M_\text{ArrayAvg, best}(\mathcal{D}) = f_\text{arr, best}(\mathcal{D})+\text{Lap}(\Delta_{f_\text{arr, best}}/\varepsilon).
\]
Clearly, both these algorithms are $\varepsilon$-DP, from Theorem \ref{thm:dp}.

Observe that since $\overline{K}\geq K$, we have that $\Delta_{f_\text{arr, best}}\leq \frac{U}{K}< \Delta_{f_\text{arr, wrap}}$. This suggests that using the \textsc{BestFit} grouping strategy yields lower estimation error, since the corresponding standard deviation of the noise added is lower by a multiplicative factor of $2$ compared to the \textsc{WrapAround} strategy. In what follows, we therefore focus our attention on the \textsc{BestFit} strategy.

We now define
\[
\tilde{\Delta}_{f_\text{arr}}:= \frac{U \, m_\text{UB}}{\sum_{\ell=1}^L \min\{m_\ell, m_\text{UB}\}}
\]
as a proxy for the upper bound $U/K$ on $\Delta_{f_\text{arr, best}}$. If $\tilde{\Delta}_{f_\text{arr}}$ is small, then so are $U/K$ and $\Delta_{f_\text{arr, best}}$.

The following lemma, that compares the \textsc{Baseline} and \textsc{Array-Averaging} algorithms, then holds:

\begin{lemma}
	For any $m_\text{UB}\leq m^\star$, we have that $\Delta_f \geq \tilde{\Delta}_{f_\text{arr}}$.
\end{lemma}
\begin{proof}
	We need only prove that $\frac{m^\star}{\sum_\ell m_\ell}\geq \frac{m_\text{UB}}{\sum_{\ell=1}^L \min\{m_\ell, m_\text{UB}\}}$, or equivalently, that $\frac{m^\star}{m_\text{UB}}\geq \frac{\sum_\ell m_\ell}{\sum_{\ell=1}^L \min\{m_\ell, m_\text{UB}\}}$. To see this, let $B:= \{\ell\in [L]: m_\ell< m_\text{UB}\}$ and $B^c:= [L]\setminus B$. Then,
	\begin{align*}
		\frac{\sum_\ell m_\ell}{\sum_{\ell=1}^L \min\{m_\ell, m_\text{UB}\}}&= \frac{\sum_{\ell\in B} m_\ell + \sum_{\ell\in B^c} m_\ell }{\sum_{\ell\in B} m_\ell + \sum_{\ell\in B^c} m_\text{UB}}\\
		&\leq \frac{\sum_{\ell\in B^c} m_\ell }{|B^c|\cdot m_\text{UB}}\\
		&\leq \frac{m^\star}{m_\text{UB}},
	\end{align*}
where the first inequality holds since $\sum_{\ell\in B^c} m_\ell \ge \sum_{\ell\in B^c} m_\text{UB}$, and the second inequality holds since $m_\ell\leq m^\star$, for all $\ell\in [L]$.
\end{proof}
%
Next, we shall embark on choosing a ``good'' $m_\text{UB}$, which provides a large ``gain" $\frac{\Delta_f}{\tilde{\Delta}_{f_\text{arr}}}$; this therefore guarantees that $\Delta_{f_\text{arr, best}}$ is small.

To this end, call $\alpha:= \frac{m^\star}{m_\text{UB}}$; in our setting, we have $\alpha\geq 1$. For fixed $\{m_\ell\}_\ell\geq 1$, observe that
\begin{align}
	\frac{\Delta_f}{\tilde{\Delta}_{f_\text{arr}}}&= \frac{\alpha}{\sum_\ell m_\ell}\cdot \sum_\ell \min\{m_\ell, m_\text{UB}\} \notag\\
	&=  \frac{\alpha}{\sum_\ell m_\ell}\cdot \sum_\ell \min\{m_\ell, {m^\star}/{\alpha}\} \notag\\
	&\leq  \frac{\alpha}{\sum_\ell m_\ell}\cdot \min\left\{ \sum_\ell m_\ell, m^\star L/\alpha\right\} = \min\left\{ \alpha, \frac{m^\star L}{\sum_\ell m_\ell}\right\}. \label{eq:sens}
\end{align}
In the above, the equality $\frac{\Delta_f}{\tilde{\Delta}_{f_\text{arr}}} = \alpha$ holds only if $m_\text{UB}\geq m_\ell$, for all $\ell \in [L]$, or equivalently, $m_\text{UB} = m^\star$. Note that in this case $\alpha = \frac{\Delta_f}{\tilde{\Delta}_{f_\text{arr}}}$ in fact equals $1$. On the other hand, the equality  $\frac{\Delta_f}{\tilde{\Delta}_{f_\text{arr}}} =\frac{m^\star L}{\sum_\ell m_\ell}\geq 1$ holds only if $m_\text{UB}\leq m_\ell$, for all $\ell\in [L]$, or equivalently, $m_\text{UB} = m_\star$. We designate $\mathsf{OPT}:= \frac{m^\star L}{\sum_\ell m_\ell}$, and note that $\Delta_f = \mathsf{OPT} \, \tilde{\Delta}_{f_\text{arr}}$ when $m_\text{UB} = m_\star$.

While the choice $m_\text{UB} = m_\star$ results in a high gain, it could potentially cause poor accuracy in estimation of the true value $f(\mathcal{D})$, due to the clipping error incurred by dropping a relatively large number of user samples. More precisely, since each array contains only very few samples,  $f_\text{arr, best}(\mathcal{D})$ could be quite different from $f(\mathcal{D})$. The next lemma shows that choosing the more favourable $m_\text{UB} = \text{med} (m_1,\ldots,m_L)$ results in a gain that is only at most a factor of 2 lower than $\mathsf{OPT}$.

\begin{lemma}
	The choice $m_\text{UB} = \text{med} (m_1,\ldots,m_L)$ results in
	\[
	\frac{\Delta_f}{\tilde{\Delta}_{f_\text{arr}}} \geq \frac{\mathsf{OPT}}{2}.
	\]
\end{lemma}
\begin{proof}
	For this choice of $m_\text{UB}$, we have that, setting $\alpha = \frac{m^\star}{m_\text{UB}}$ and $B:= \{\ell\in [L]: m_\ell< m_\text{UB}\}$,
	\begin{align*}
		\frac{\Delta_f}{\tilde{\Delta}_{f_\text{arr}}}&= \frac{\alpha}{\sum_\ell m_\ell}\cdot \sum_\ell \min\{m_\ell, m_\text{UB}\}\\
		&= \frac{\alpha}{\sum_\ell m_\ell}\cdot \left[\sum_{\ell\in B} m_\ell + \frac{m^\star\left\lceil L/2\right\rceil}{\alpha} \right]\\
		&\geq \frac{\mathsf{OPT}}{2},
	\end{align*}
since $m_\ell\geq 0$, for all $\ell\in [L]$.
\end{proof}
In our experiments in Section \ref{sec:results}, we employ $m_\text{UB} = \text{med} (m_1,\ldots,m_L)$ in demonstrating the empirical performance of \textsc{Array-Averaging} under the \textsc{BestFit} grouping strategy.
\subsection{\textsc{Levy}}
In this algorithm and the next, we attempt to reduce $\Delta_f$ further by simultaneously clipping the number of samples per user and the range of speed values. The algorithm that we present in this section puts together the \textsc{Array-Averaging} algorithm presented in the previous subsection and Algorithm 1 in \cite{userlevel}. However, unlike in the previous section, we make use of a value of $m_\text{UB}$ that is different from $\text{med} (m_1,\ldots,m_L)$ and is instead tailored, via heuristic calculations, to the sensitivity of the function associated with this algorithm.

We now present a qualitative description of the algorithm; the exact choice of the parameters is based on certain heuristics that will be described later.

As in the previous section, we first clip the number of samples contributed by any user using (one of) the array-based grouping strategies and treat the thus-formed arrays as pseudo-users. More precisely, let $A_1,\ldots,A_{\overline{K}}$ be the arrays obtained using the \textsc{BestFit} grouping strategy with a certain value of $m_\text{UB}$ to be specified later (recall that we use \textsc{BestFit} instead of \textsc{WrapAround} since under the \textsc{BestFit} strategy, a user can influence at most one array, unlike in the \textsc{WrapAround} strategy). Next, we clip the range of speed values. Suppose that $a, b$ with $0\leq a\leq b\leq U$ are real numbers such that the speed values of the users lie in $[a,b]$ with high probability. The \textsc{Levy} algorithm first privately estimates this interval $[a,b]$, with privacy loss set to $\varepsilon/2$. We then define the function
\[
f_\text{Levy}(\mathcal{D}):= \frac{1}{\overline{K}}\cdot \sum_{i=1}^{\overline{K}} \Pi_{[a,b]} (\overline{A_i}),
\]
where $\Pi_{[a,b]}$ denotes the projection operator into the interval $[a,b]$, with
\[
\Pi_{[a,b]}(x) = \min \{b, \max\{a,x\}\},
\]
for any $x\in \mathbb{R}$. As before, $\overline{A_i}:= \frac{1}{w(A_i)} \sum_{j=1}^{w(A_i)} A_i(j)$ is the mean of the samples contributed by array $i\in [\overline{K}]$. Note that now the user-level sensitivity is given by
\begin{equation}
\Delta_{f_\text{Levy}} = \frac{b-a}{\overline{K}},
\label{eq:sensitivitylevy}
\end{equation}
since we implicitly assume that grouping is performed according to the \textsc{BestFit} strategy.

The \textsc{Levy} algorithm (with \textsc{BestFit} grouping) then computes
\[
M_\text{Levy}(\mathcal{D}) = f_\text{Levy}(\mathcal{D})+\text{Lap}(2\Delta_{f_\text{Levy}}/\varepsilon).
\]
Note that in the above expression, the privacy loss is assumed to be $\varepsilon/2$. Overall, the privacy loss for both private interval estimation and for private mean estimation is $\varepsilon$, following the basic composition theorem \cite[Corollary 3.15]{dworkroth}. Hence, the algorithm \textsc{Levy} is $\varepsilon$-DP.

Observe that when $b-a<U/2$, the standard deviation of the noise added in this case is less than that added using \textsc{Array-Averaging} with \textsc{BestFit} grouping.

We now explain the heuristics we employ to select the parameters $a, b$ in the \textsc{Levy} algorithm and $m_\text{UB}$.

\subsubsection{Private Interval Estimation}
The subroutine we use in this algorithm to privately compute the ``high-probability'' interval $[a,b]$ is borrowed from Algorithm 6 in \cite{userlevel}, which is used to privately compute an estimate of the $\left(\frac14,\frac34\right)$-interquantile interval of a given collection of i.i.d. (random) scalar values $X_1,\ldots,X_n$ (in our setting, these are the array means $\overline{A}_1,\ldots,\overline{A}_{\overline{K}}$). In order to describe this subroutine, which incorporates some additional modifications to suit our setting, we provide necessary background. The application of Algorithm 6 in \cite{userlevel} crucially relies on the following concentration property of the values $X_1,\ldots,X_n$ provided to the algorithm:
\begin{definition}
	A random sequence $X^n$ supported on $[0,M]$ is $(\tau, \gamma)$-concentrated ($\tau$ is called the ``concentration radius”) if there exists $x_0\in [0,M]$ such that with probability at least $1 - \gamma$, 
	\[
	\max_{i\in [n]} |X_i-x_0|\leq \tau.
	\]
\end{definition}

First, observe that each  ``array mean'' $\overline{A}_i$, $i\in [\overline{K}]$ is the sample mean of values $A_i(j)\in [0:U]$, where $1\leq j\leq w(A_i)$. {Although the samples in each array are drawn in a potentially non-i.i.d. fashion from an unknown distribution, we simply rely on heuristics that assume that the data samples are i.i.d. and that each array is fully filled, with $\overline{K} = K$.}


By an application of Hoeffding's inequality (see, e.g.,  \cite[Theorem 2.2.6]{vershynin}), we have that each array mean $\overline{A}_i$, $i\in [\overline{K}]$, is $\frac{U^2}{4m_\text{UB}}$-subGaussian (see, e.g., \cite[Theorem 2.1.1]{ed}). Hence, if the speed samples were i.i.d., from a simple application of the union bound (see, e.g., \cite[Theorem 2.2.1]{ed}), we obtain that the sequence $(\overline{A}_1,\ldots,\overline{A}_{\overline{K}})$ is in fact $(\tau,\gamma)$-concentrated about the expected value, where
\begin{equation}
\tau = U\cdot \sqrt{\frac{\log (2\overline{K}/\gamma)}{2m_\text{UB}}}.
\label{eq:tau}
\end{equation}
We use this value of $\tau$ to compute the ``high-probability'' interval $[a,b]$, with privacy loss $\varepsilon/2$, as shown in Algorithm \ref{alg:interval}.
\begin{algorithm}
	\caption{Private Interval Estimation}
	\begin{algorithmic}[1]
		\Procedure{PrivateInterval}{$(\overline{A}_1,\ldots,\overline{A}_{\overline{K}}), \varepsilon,\tau,U$}
		\State Divide $[0,U]$ into $U/\tau$ disjoint bins, each of width $\tau$. Let $T$ be the midpoints of the bins.
		\State Set $\mu_i\gets \min_{x\in T} |\overline{A}_i-x|$, for all $i\in [\overline{K}]$.
		\State For $x\in T$, define the cost function
		\[
		c(x) = \max \{|\{i\in [\overline{K}]: \mu_i<x\}|, |\{i\in [\overline{K}]: \mu_i>x\}|\}.
		\]
		\State Sample $\hat{\mu}\in T$ from $Q$ where
		\[
		Q(\hat{\mu}=x) = \frac{e^{-\varepsilon c(x)/4}}{\sum_{x'\in T}e^{-\varepsilon c(x')/4}}.
		\]
		\State Set $a\gets \max\left\{0,\hat{\mu}-\frac{3\tau}{2}\right\}$, $b\gets \min \left\{\hat{\mu}+\frac{3\tau}{2},U\right\}$.
		\EndProcedure
	\end{algorithmic}
	\label{alg:interval}
\end{algorithm}
\subsubsection{Choosing $m_\text{UB}$}
The subroutine we use to choose the length of the arrays $m_\text{UB}$ is tailored to the sensitivity of the function $f_\text{Levy}$ in \eqref{eq:sensitivitylevy}. For a fixed $m=m_\text{UB}$, note that from Algorithm \ref{alg:interval},
\begin{align}
	\label{eq:deltalevy}
\Delta_{f_\text{Levy}}(m)
&= \min\left\{\frac{3\tau}{\overline{K}}, \frac{U}{\overline{K}}\right\} \notag\\ &= \min\left\{\frac{3U}{\overline{K}}\sqrt{\frac{\log (2\overline{K}/\gamma)}{2m}}, \frac{U}{\overline{K}}\right\}.
\end{align}
In order to reduce the sensitivity, our heuristic aims to maximize the $\overline{K} \sqrt{m}$ term  in \eqref{eq:deltalevy}, and sets
\begin{equation}
	\label{eq:uboptim}
m_\text{UB} \in \arg \max_m \frac{\sum_{\ell=1}^{L}\min \left \{m_{\ell}, m \right \}}{\sqrt{m}},
\end{equation}
where the right-hand side above is a set of cardinality potentially larger than $1$. We then numerically solve this optimization problem over $m_\star\leq m\leq m^\star$, to obtain a single, suitable value of $m_\text{UB}$.

\subsection{\textsc{Quantile}}
\label{sec:quantile}
The \textsc{Quantile} algorithm is a ``noisy mean-of-projections" that is quite similar in flavour to the \textsc{Levy} algorithm, with the only difference being the choice of the interval $[a',b']$ into which the array means are projected.

{ More precisely, the \textsc{Quantile} algorithm sets $m_\text{UB}$ as in \eqref{eq:uboptim}, and uses the \textsc{BestFit} grouping strategy to create arrays containing the speed samples.} Next, the algorithm privately estimates a ``high-probability'' interval $[a',b']$ (in a manner different from Algorithm \ref{alg:interval}), with privacy loss $\varepsilon/2$. We then define the function
\[
f_\text{Quantile}(\mathcal{D}):= \frac{1}{\overline{K}}\cdot \sum_{i=1}^{\overline{K}} \Pi_{[a',b']} (\overline{A_i}),
\]
where the projection operator and the array means are as defined earlier. Once again, we have that the sensitivity
\[
\Delta_{f_\text{Quantile}} = \frac{b'-a'}{\overline{K}},
\]
The \textsc{Quantile} algorithm (with \textsc{BestFit} grouping) then computes
\[
M_\text{Quantile}(\mathcal{D}) = f_\text{Quantile}(\mathcal{D})+\text{Lap}(2\Delta_{f_\text{Quantile}}/\varepsilon),
\]
where the privacy loss for the mean estimation is set to $\varepsilon/2$. We now describe two subroutines that we use to privately estimate the high-probability $[a',b']$.
\subsubsection{\textsc{FixedQuantile}}
This subroutine privately estimates the $\left(\frac{1}{10},\frac{9}{10}\right)$-interquantile interval of the array means $\overline{A}_1,\ldots,\overline{A}_{\overline{K}}$, using Algorithm 2 in \cite{smith}. A notable difference between this algorithm and Algorithm 6 in \cite{userlevel} (used in \textsc{Levy}) is that this algorithm does not quantize the interval $[0,U]$ in which the samples lie as in Step 2 of Algorithm \ref{alg:interval}.  Note that this algorithm computes estimates of $a'$ and $b'$ separately, and we set the privacy loss of each of these computations to be $\varepsilon/4$, so that the overall privacy loss for private quantile estimation is $\varepsilon/2$.
\subsubsection{\textsc{OptimizedQuantile}}
This subroutine privately estimates (with privacy loss $\varepsilon/2$) the interval $[a',b']$, which is now chosen to minimize the sum of the absolute estimation errors due to privacy and due to clipping the speed values, following the work in \cite{amin}. In particular, an argument similar to that in Section 3 in \cite{amin} advocates that in order to minimize this sum of absolute estimation errors, we need to set $b$ to be the $\left \lceil \frac{2}{\varepsilon}\right \rceil ^{\text{th}}$-largest value among $(\overline{A}_1,\ldots,\overline{A}_{\overline{K}})$, and by symmetry, we need to set $a$ to be the $\left \lfloor \frac{2}{\varepsilon}\right \rfloor ^{\text{th}}$-smallest value among $(\overline{A}_1,\ldots,\overline{A}_{\overline{K}})$. We then proceed to privately estimate the $\left(\frac{1}{\overline{K}}\cdot\left \lceil \frac{2}{\varepsilon}\right \rceil ,1-\frac{1}{\overline{K}}\cdot\left \lceil \frac{2}{\varepsilon}\right \rceil\right)$-interquantile interval using Algorithm 2 in \cite{smith}. Again, we set the privacy loss of computing $a'$ and $b'$ to be individually $\varepsilon/4$, so that the overall privacy loss for private quantile estimation is $\varepsilon/2$.

%% file: results.tex
\label{sec:results}
\subsection{Setup}

\begin{figure*}[]
	\centering
	\includegraphics[width=\linewidth]{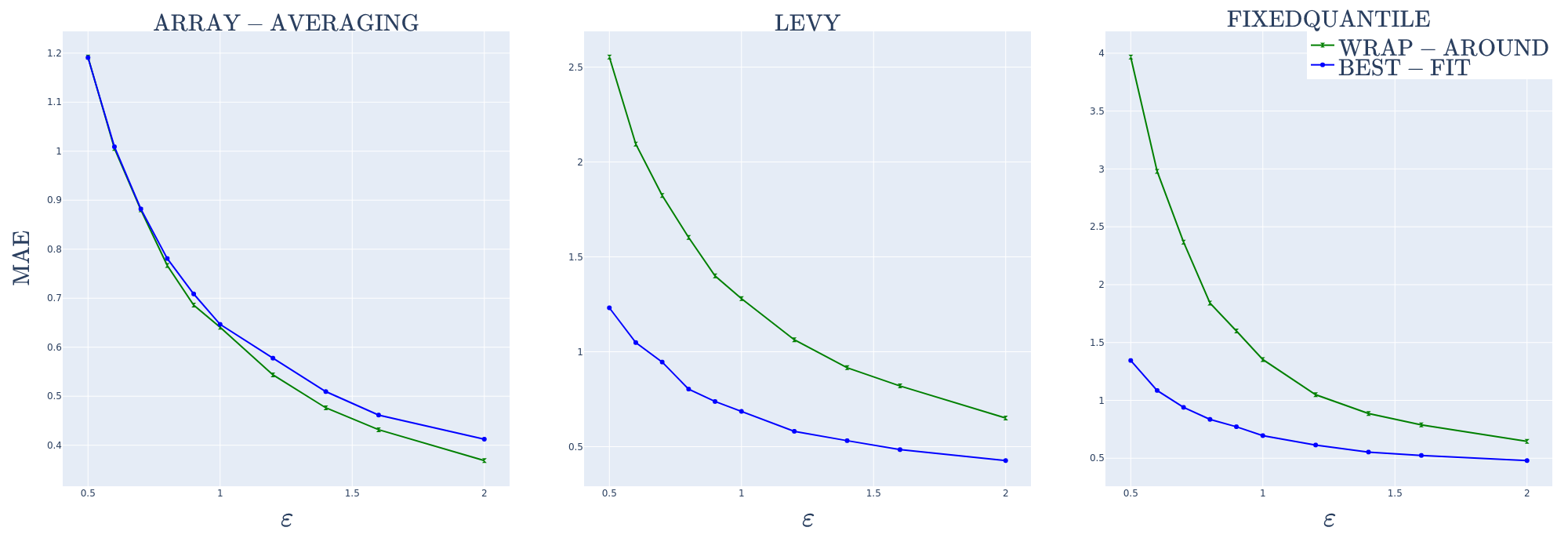}
	\caption{Plots comparing $E_\text{MAE}$ for the \textsc{Array-Averaging}, \textsc{Levy}, and \textsc{Quantile} algorithms, under the two grouping strategies. For the \textsc{Levy} algorithm, we fix $\gamma = 0.2$, and for the \textsc{Quantile} algorithm, we use the \textsc{FixedQuantile} subroutine.}
	\label{fig:group}
\end{figure*}

{We test and evaluate our algorithms on two types of datasets: 1) a real-world ITMS dataset $\mathcal{C}$ containing non-i.i.d. speed values contributed by bus drivers over a span of 3 days in an Indian city, and 2) a synthetic dataset $\mathcal{D}$ containing i.i.d. samples drawn using insights from the ITMS data. We mention that we first preprocess the dataset $\mathcal{C}$ to remove the $0$-valued speed samples, which correspond to times when the buses are stationary. Further, towards reducing the error due to clipping, we populate the \emph{sample means} of the speed samples of every user in the arrays created, in place of the true speeds. More precisely, suppose that locations $[j,j+\min\{m_\ell,m_\text{UB}\}-1]$ in array $A_i$ are allotted to samples from user $\ell\in [L]$, using either the \textsc{WrapAround} or \textsc{BestFit} strategy. The speed samples that are populated in these locations are $(\mu_\ell,\ldots,\mu_\ell) \in (0,U]^{\min\{m_\ell,m_\text{UB}\}}$, where $\mu_\ell:= \frac{1}{m_\ell}\sum_{t=1}^{m_\ell}S_t^{(\ell)}$. 
	
We ran each private mean estimation algorithm described earlier $10^4$ times; the plots showing the performance of our algorithms are shown in this section. We test the accuracy of our algorithms, via the error in estimation, for privacy loss $\varepsilon$ ranging from $0.5$ to $2$.}

For any dataset $\overline{D}$, we use the mean absolute error (or MAE) metric, $$E_\text{MAE}(\overline{\mathcal{D}}) = {\frac{\sum_{i=0}^{10^4} \left \lvert M^{(i)}(\overline{\mathcal{D}}) - \mu(\overline{\mathcal{D}})\right \lvert}{10^4}},$$ to evaluate the performance of the algorithm. Here, for $i\in [10^4]$, $M^{(i)}(\overline{\mathcal{D}})$ is the result of running the algorithm $M$ (which is one of  \textsc{Array-Averaging}, \textsc{Levy}, and \textsc{Quantile}, with a choice of subroutines) on the dataset $\overline{\mathcal{D}}$ at iteration $i$, and $\mu(\overline{\mathcal{D}})$ is the true sample mean.  Since, for $Z\sim \text{Lap}(b)$, we have $\E[|Z|] = b$, we simply use $E_\text{MAE}(\overline{\mathcal{D}}) = \Delta_f/\varepsilon$ for \textsc{Baseline}.


\subsection{Experimental Results for ITMS Data}

For our ITMS dataset $\mathcal{C}$, the number of users $L=217$, the maximum and minimium number of samples contributed by any user are respectively $m^{*} = 417$ and $m_* = 1$, and the total number of samples contributed by all user was $\sum_\ell m_\ell = 17166$. The true sample mean of speed records $\mu(\mathcal{C})$ was $20.66769$. We also mention that we pick $\text{med}(m_1,\ldots,m_L) = 46$. {Also, for  $m_\text{UB}$ as in \eqref{eq:uboptim}, the number of arrays (or pseudo-users) $K$ under \textsc{WrapAround} grouping was $160$ and the number of arrays $\overline{K}$ under  \textsc{BestFit} grouping was $164$.}


First, we compare the performance of the algorithms that employ clipping of the number of user samples, under \textsc{WrapAround} and \textsc{BestFit} grouping strategies. In particular, we evaluate $E_\text{MAE}$ for \textsc{Array-Averaging}, for \textsc{Levy} with $\gamma = 0.2$, and for \textsc{Quantile}, with the \textsc{FixedQuantile} subroutine, under the two grouping strategies. Figure \ref{fig:group} shows these comparisons, for $\varepsilon$ varying from $0.5$ to $2$. While for the \textsc{Array-Averaging} algorithm, the performance of the two grouping strategies is similar, it is clear from the plots that for the \textsc{Levy} and \textsc{Quantile} algorithms, the \textsc{BestFit} strategy performs much better than the \textsc{WrapAround} strategy. This conforms with our expectation (see the discussion in Section \ref{sec:grouping}, since the \textsc{WrapAround} strategy results in a multiplicative factor of $2$ in the user-level sensitivity).

In Figure \ref{fig:itmscompare}, we hence fix the grouping strategy to be \textsc{BestFit} and compare the performance of the 4 algorithms described in Section \ref{sec:alg} on the non-i.i.d. ITMS dataset. Clearly, clipping of the number of samples per user and the speed values results in much better performance than the na\"{i}ve \textsc{Baseline} approach. However, there is very little difference between the performance of the other algorithms, possibly owing to the relatively small size of the ITMS dataset.


%
%

%
%
%

\subsection{Experimental Results for Synthetic Data}

Next, we artificially generate data samples, by drawing insights from the ITMS dataset. This artificially generated dataset, denoted as $\mathcal{D}$, allows us to study the performance of the algorithms in Section \ref{sec:alg} when the number of samples per user or the number of users is large. In particular, let $\widehat{L}$ denote the number of users and $\widehat{m_\ell}$, for $\ell \in [\widehat{L}]$ denote the number of samples contributed by user $\ell$. Recall that $L$ and $\{m_\ell: \ell \in [L]\}$ denoted the number of users and the collection of the number of samples per user for the original ITMS dataset. We generate i.i.d. speed samples $\left\{\widehat{S}_j^{(\ell)}: \ell \in [\widehat{L}], j\in [\widehat{m_\ell}]\right\}$, where for any $\ell, j$, we have $\widehat{S}_j^{(\ell)}\sim \Pi_{[0,U]}(X)$, where $X\sim \mathcal{N}(\mu(\mathcal{C}),\sigma^2(\mathcal{C}))$. Here, recall that $\mu(\mathcal{C})$ is the true sample mean of the ITMS dataset. Further, $\sigma^2(\mathcal{C})$ is the true sample variance of the ITMS dataset, defined as
\[
\sigma^2(\mathcal{C}) = \frac{1}{\sum_{\ell=1}^{L}{m_\ell}}\cdot \sum_{\ell=1}^L \sum_{j=1}^{m_\ell}\left(S_j^{(\ell)}-\mu(\mathcal{C})\right)^2.
\]
For the ITMS dataset, $\sigma^2(\mathcal{C}) = 115.135$. Further, note that the  random variable $\widehat{S}_j^{(\ell)}$ obeys:
\begin{align*}
\Pr\left[\widehat{S}_j^{(\ell)} = 0\right] &= Q\left(\frac{\mu(\mathcal{C})}{\sigma(\mathcal{C})}\right)\ \text{and}\\ \Pr\left[\widehat{S}_j^{(\ell)} = U\right] &= Q\left(\frac{U-\mu(\mathcal{C})}{\sigma(\mathcal{C})}\right).
\end{align*}
  In what follows, we set $\mu := \mu(\mathcal{C})$ and $\sigma^2 := \sigma^2(\mathcal{C})$. We consider two settings of interest, for a fixed positive integer $\lambda$.
\begin{enumerate}
	\item \textbf{Sample scaling}: In this setting, we set $\widehat{L} = L$ and $\widehat{m}_\ell = 10\cdot m_\ell$, for all $\ell\in [\widehat{L}]$. In this case, $\widehat{m}^\star := \max_{\ell\in [\widehat{L}]} \widehat{m}_\ell = \lambda\cdot m^\star$. 
	\item \textbf{User scaling}: In this setting, we set $\widehat{L} = \lambda\cdot L$ where $\lambda$ is a positive integer. For each $i\in [\lambda]$, we let $\widehat{m}_{\lambda\cdot(\ell-1)+i} = m_\ell$,  for all $\ell\in [{L}]$. Here, $\widehat{m}^\star = m^\star$.
\end{enumerate}

We also let $\widehat{m}_\text{UB}$ denote the new lengths of the arrays used for the dataset $\mathcal{D}$, with $\widehat{K}$ and $\widehat{\overline{K}}$ denoting, respectively, the resultant number of arrays under the \textsc{WrapAround} and \textsc{BestFit} grouping strategies. Further, for each of the algorithms considered in Section \ref{sec:alg}, we denote their sensitivities under sample scaling (resp. user scaling) by using an additional superscript `$(s)$' (resp. `$(u)$') over the existing notation.

We first consider the setting of \textbf{sample scaling}, with $\lambda = 10$. Figure \ref{fig:samplescale} shows comparisons between the different algorithms on this synthetic dataset, with $\gamma = 0.2$ for \textsc{Levy}. Here, we implicitly assume that the grouping strategy employed is \textsc{BestFit}. It is evident from the plots that \textsc{Baseline} performs much poorer than the clipping algorithms, as expected. However, interestingly, the \textsc{Levy} algorithm performs better than all other algorithms. 

To see why, we first need the following simple lemma.

\begin{lemma}
	\label{lem:samplescale}
	Under sample scaling, we have that 
	\[
	\widehat{m}_\text{UB} = \lambda\cdot m_\text{UB},\ \widehat{K} = K,\ \text{and}\ \widehat{\overline{K}} = \overline{K},
	\]
	for $m_\text{UB}$ chosen either as the sample median or as in \eqref{eq:uboptim}.
\end{lemma}
We now proceed to justify the performance trends observed under sample scaling, by analyzing the dependence of the sensitivities on the scaling factor.

\begin{proposition}
	We have that
	\begin{align*}
	\Delta^{(s)}_{f_\text{Baseline}} = \Delta_{f_\text{Baseline}},\
	\Delta^{(s)}_{f_\text{arr, wrap}} = \Delta_{f_\text{arr, wrap}},\ \Delta^{(s)}_{f_\text{arr, best}} = \Delta_{f_\text{arr, best}},
	\end{align*}
and 
\[
\Delta^{(s)}_{f_\text{Levy}} = \frac{1}{\sqrt{\lambda}}\cdot  \Delta_{f_\text{Levy}}.
\]
\end{proposition}
\begin{proof}
	The above scaling of the sensitivities under sample scaling follows directly from the definitions of the sensitivities and from Lemma \ref{lem:samplescale}.
\end{proof}

It is clear from the above proposition that the sensitivity (and hence the amount of Laplace noise added for privacy) is much smaller for \textsc{Levy} as compared to \textsc{Baseline} and \textsc{Array-Averaging}, for large enough $\lambda$. 

Now, for the case of the \textsc{Quantile} algorithm, a precise analysis of the sensitivity of the algorithm as stated is hard, owing to possible errors in the private estimation of the quantiles used.
We believe the relative better performance of \textsc{Levy} is due the additional errors in private estimation of the quantiles in \textsc{Quantile}. Hence, in general, when the dataset contains a large number of samples per user, we recommend using \textsc{Levy}.

Figure \ref{fig:userscale} for \textbf{user scaling} shows that the \textsc{FixedQuantile} subroutine outperforms all the other algorithms. Furthermore, \textsc{Array-Averaging} performs second-best. We mention that here we use $m_\text{UB}$ as in \eqref{eq:uboptim} for \textsc{Array-Averaging} as well, to maintain uniformity. We now justify these performance trends. Analogous to the setting with sample scaling, the following simple lemma holds.
\begin{figure}[!h]
	\centering
	\includegraphics[width=0.9\linewidth]{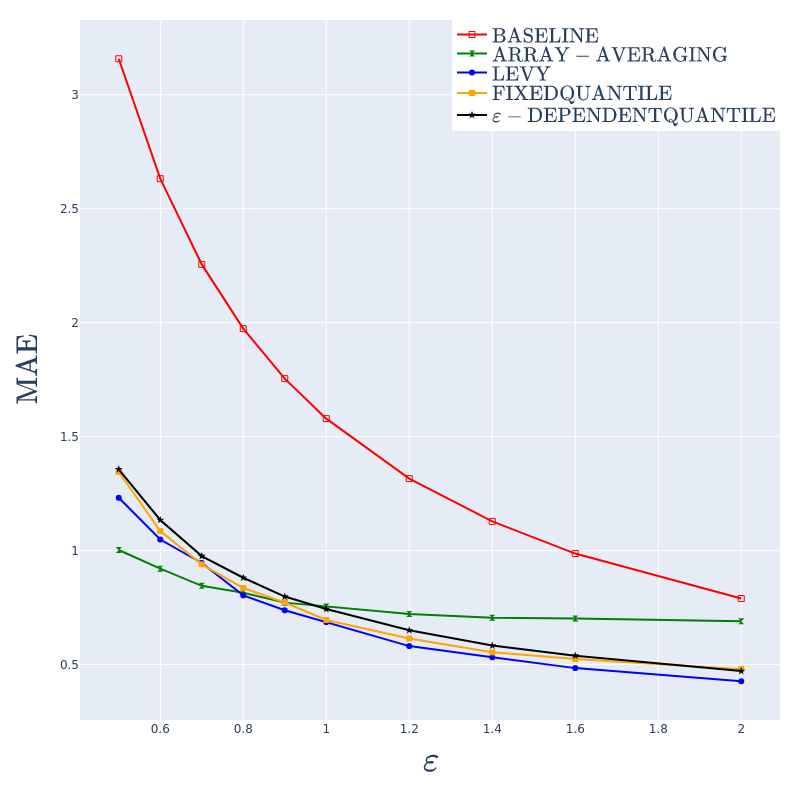}
	\caption{Plots comparing the performance of algorithms on real-world ITMS data}
	\label{fig:itmscompare}
\end{figure}



\begin{figure}[!h]
	\centering
	\includegraphics[width=0.9\linewidth]{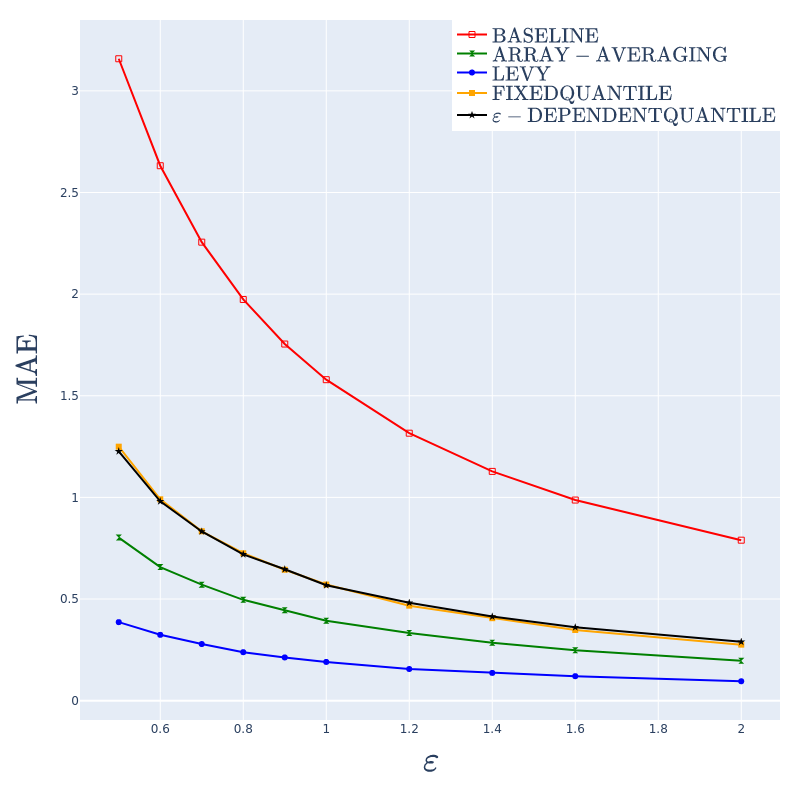}
	\caption{Plots comparing the performance of algorithms under \textbf{sample scaling}}
	\label{fig:samplescale}
\end{figure}

\begin{figure}[!h]
	\centering
	\includegraphics[width=0.9\linewidth]{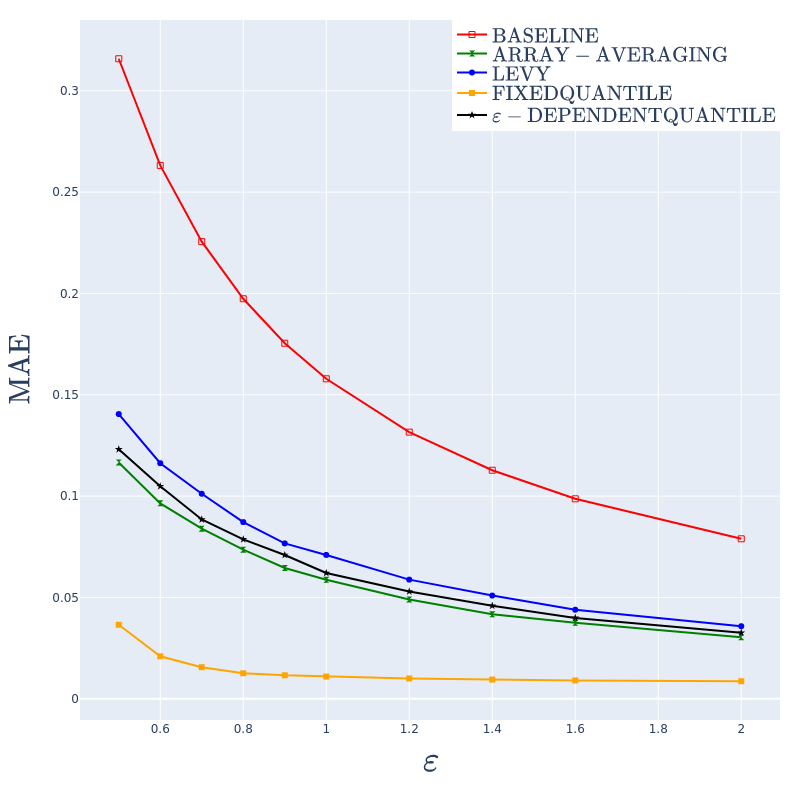}
	\caption{Plots comparing the performance of algorithms under \textbf{user scaling}}
	\label{fig:userscale}
\end{figure}
\begin{lemma}
	\label{lem:user_scale}
	Under user scaling, we have that 
	\[
	\widehat{m}_\text{UB} = m_\text{UB},\ \widehat{K} = \lambda\cdot K,\ \text{and}\ \widehat{\overline{K}} = \lambda\cdot \overline{K},
	\]
	for $m_\text{UB}$ chosen as in \eqref{eq:uboptim}.
\end{lemma}
With the above lemma in place, we first provide a justification of the superior performance of \textsc{Array-Averaging}, compared to \textsc{Levy} and \textsc{$\varepsilon$-DependentQuantile}, where we assume that the exact sample-dependent quantiles are computed in \textsc{$\varepsilon$-DependentQuantile} (i.e., we assume that the error in estimation of the sample quantiles due to privacy is zero). To  denote the standard deviation of noise added in each algorithm described in Section \ref{sec:alg}, under user scaling, we use the same notation as for the sensitivities under user scaling, but replacing $\Delta$ with $\sigma$. We recall that for $X \sim \mathrm{Lap}(b)$, the standard deviation of $X$, denoted by $\sigma$, is $\sqrt{2} b$. We implicitly assume that \textsc{BestFit} is used as the grouping strategy in each algorithm.
\begin{proposition}
 For any fixed $\varepsilon$, for large enough $\lambda$, we have that with high probability, $$\sigma^{(u)}_{f_\text{arr, best}}< \min\left\{ \sigma^{(u)}_{f_\text{Levy}}, \sigma^{(u)}_{f_\text{Quantile}}\right\},$$ under the \textsc{$\varepsilon$-DependentQuantile} subroutine for \textsc{Quantile}, if the exact sample quantiles are employed.
\end{proposition}

\begin{proof}
	First consider \textsc{Array-Averaging} under user scaling. Recall that $\Delta^{(u)}_{f_\text{arr, best}} = \frac{U}{\widehat{\overline{K}}} = \frac{U}{\lambda\cdot {\overline{K}}}$, where the last equality follows from Lemma \ref{lem:user_scale}. Hence, we have that $\sigma^{(u)}_{f_\text{arr, best}} = \frac{\sqrt{2}U}{\lambda{\overline{K}}\varepsilon}$.
	
	Now, consider \textsc{Levy} under user scaling. From \eqref{eq:deltalevy}, we have that 
	\[
	\sigma_{f_\text{Levy}}^{(u)}
	= \min\left\{\frac{6\sqrt{2}\cdot\tau}{\widehat{\overline{K}}\cdot \varepsilon}, \frac{2\sqrt{2}\cdot U}{\widehat{\overline{K}}\cdot \varepsilon}\right\} = \frac{2\sqrt{2}\cdot U}{\lambda{\overline{K}}\varepsilon},
	\]
	for $\lambda\geq \frac{\gamma}{2\overline{K}}\cdot e^{2m_\text{UB}/9}$ (see \eqref{eq:tau}). Here, in the last equality above, we once again invoke Lemma \ref{lem:user_scale} to obtain $$\widehat{\tau} =  U\cdot \sqrt{\frac{\log (2\lambda\overline{K}/\gamma)}{2m_\text{UB}}}.$$ Hence, for $\lambda\geq \frac{\gamma}{2\overline{K}}\cdot e^{2m_\text{UB}/9}$, we have that $\sigma^{(u)}_{f_\text{arr, best}}< \sigma^{(u)}_{f_\text{Levy}}$.
	
	Now consider the \textsc{$\varepsilon$-DependentQuantile} under user scaling. Fix some $\delta>0$ such that $\delta<U/4$. We claim that for large enough $\lambda$, the true $\left({\widehat{\overline{K}}}^{-1}\cdot\left \lceil \frac{2}{\varepsilon}\right \rceil ,1-{\widehat{\overline{K}}}^{-1}\cdot\left \lceil \frac{2}{\varepsilon}\right \rceil\right)$-interquantile interval of the samples contains $[\delta,U-\delta]$, with high probability. To this end, let $t := \left \lceil \frac{2}{\varepsilon}\right \rceil$ and let $B$ denote the event that $\left({\widehat{\overline{K}}}^{-1}\cdot{t},1-{\widehat{\overline{K}}}^{-1}\cdot{t}\right)$-interquantile interval is contained in $[\delta,U-\delta]$. We shall now show that $\Pr[B]$ decays to zero as $\lambda$ increases to infinity. Observe that
	\begin{align}
	&\Pr\left[B\right] \notag\\
	&\leq \Pr\left[\exists\ {\widehat{\overline{K}}}-2t \text{ array means all lying in } [\delta,U-\delta]\right]. \label{eq:inter}
	\end{align}
Now, let $B_i$ denote the event that $\widehat{\overline{A}}_i$, for $i\in [\widehat{\overline{K}}]$ lies in $[\delta,U-\delta]$, and let $i^\star$ denote that array index $i$ that maximizes $\Pr[B_i]$. Note that 
\[
\Pr[B_{i^\star}]\leq 1-\max\left\{\Pr[\widehat{\overline{A}}_{i^\star} = 0], \Pr[\widehat{\overline{A}}_{i^\star} = U]\right\}.
\]

Further, we have that
\begin{align*}
\Pr[\widehat{\overline{A}}_{i^\star} = 0] &= \left(\Pr[\widehat{S}_1^{(1)} = 0]\right)^{w(\widehat{A}_{i^\star})}\\
&= (1-Q(\mu/\sigma))^{w(\widehat{A}_{i^\star})}>0.
\end{align*}
where we use the fact that the samples $\widehat{S}_j^{(\ell)}$, $\ell\in [\widehat{L}]$, $j\in [\widehat{m}_\ell]$ are drawn i.i.d. according to the ``projected'' Gaussian distribution described earlier. Since $w(\widehat{A}_{i^\star})>0$, we get that $\Pr[\widehat{\overline{A}}_{i^\star} = 0]>0$ and hence that $\Pr[B_{i^\star}]<1$.

Employing a union bound argument to upper bound the probability in \eqref{eq:inter}, we get that
\begin{align*}
	\Pr[B]&\leq \binom{{\widehat{\overline{K}}}}{{\widehat{\overline{K}}}-2t}\cdot \prod_{i=1}^{\widehat{\overline{K}}}\Pr\left[B_i\right]\\
	&\leq \binom{{\widehat{\overline{K}}}}{{\widehat{\overline{K}}}-2t}\cdot \big(\Pr\left[B_{i^\star}\right]\big)^{\widehat{\overline{K}}} = \binom{{\widehat{\overline{K}}}}{2t}\cdot \big(\Pr\left[B_{i^\star}\right]\big)^{\widehat{\overline{K}}}
\end{align*}
where the second inequality holds by the definition of $i^\star$.
Since $\binom{{\widehat{\overline{K}}}}{2t}$ grows polynomially in $\lambda$, for a fixed $t$ (see Lemma \ref{lem:user_scale}), we have that for any fixed $\delta>0$, and for a fixed, small $\beta\in (0,1)$, there exists $\lambda_0$, such that for all $\lambda\geq \lambda_0$, we have that $\Pr[B]<\beta$.


Therefore, for any arbitrarily small $\beta > 0$ and for $\lambda\geq \lambda_0$, we have that with probability at least $1-\beta$,
\[
\sigma^{(u)}_{f_\text{Quantile}}
\geq \frac{2\sqrt{2}\cdot(U-2\delta)}{\widehat{\overline{K}}\cdot \varepsilon}> \frac{\sqrt{2}\cdot U}{\widehat{\overline{K}}\cdot \varepsilon} =  \sigma^{(u)}_{f_\text{arr, best}},
\]
where the second inequality holds since $\delta< U/4$.

Hence, by picking $\lambda\geq \max\left\{\lambda_0,\frac{\gamma}{2\overline{K}}\cdot e^{2m_\text{UB}/9}\right\}$, we get $\sigma^{(u)}_{f_\text{arr, best}}< \min\left\{ \sigma^{(u)}_{f_\text{Levy}}, \sigma^{(u)}_{f_\text{Quantile}}\right\},$ with probability at least $1-\beta$.
\end{proof}

All that remains is to show that the \textsc{FixedQuantile} algorithm performs better than \textsc{Array-Averaging}. By arguments similar to those above, we observe that since the \textsc{FixedQuantile} subroutine eliminates a much larger fraction of data samples in the computation of $[a',b']$, with high probability, we have that $b'-a'$ is small, and hence the amount of noise added in this case is lower than that for \textsc{Array-Averaging}. Hence, in general, for datasets with a large number of users, we recommend using \textsc{FixedQuantile}.

Until now, we have worked with either real-world (ITMS) datasets or with datasets with i.i.d. speed samples, both of which exhibit typical behaviour in that most of the speed samples take non-extreme values, with high probability. In the sections that follow, we analyze ``worst-case'' datasets that maximize the error in estimation, for the general setting of pseudo-user creation-based algorithms that clip the number of samples contributed by users. For this setting, we provide explicit bounds on worst-case error and choices for the lengths of arrays in order to minimize the worst-case error.

%% file: minimax.tex
\label{sec:minimax}
In this section, we present a ``minimax'' characterization of the \emph{total} estimation error due to clipping and noise in an algorithm that is based on the creation of arrays, or pseudo-users; in other words, we present a characterization of the total  estimation error of \textsc{Array-Averaging} under the ``best'' choice of $m_\text{UB}$, on the ``worst'' dataset. As an immediate corollary, we obtain an upper bound on the total estimation error of \textsc{Array-Averaging} under the ``best'' choice of $m_\text{UB}$ on \emph{any} dataset. We emphasize, though, that this choice of $m_\text{UB}$ is only \emph{minimax} optimal and not necessarily optimal in a typical (or \emph{average-case}) setting. We mention that a general minimax theory for balancing errors due to estimation and privacy was left open in \cite{wassermanzhou} (see also \cite{minimaxopt}).

In what follows, we assume that the arrays $A_i$, $i\in [\overline{K}]$ are fully filled and use $\overline{K} = \frac{\sum_{\ell=1}^{L}\min \left \{m_{\ell}, m_\text{UB} \right \}}{m_\text{UB}}$ in all subsequent equations. For ease of reading, we set $m=m_\text{UB}$.

For a given dataset $\mathcal{D}'$, following the definition of the \emph{empirical loss} corresponding to a mechanism in \cite[Sec. 1.1]{primer_stat}, we let $E^{(\varepsilon)}(\mathcal{D}',m)$ to be the overall error, due to clipping and due to noise addition, incurred by the \textsc{Array-Averaging} algorithm, where
\[
E^{(\varepsilon)}(\mathcal{D}',m) = \left \lvert f_\text{arr}(\mathcal{D}') - f(\mathcal{D}')\right\rvert + \tilde{\Delta}_{f_\text{arr}}/\varepsilon.
\]
Note that in the above expression, if $Z\sim \text{Lap}(\tilde{\Delta}_{f_\text{arr}}/\varepsilon)$ is the random variable that represents the noise added for privacy, we have that $\tilde{\Delta}_{f_\text{arr}}/\varepsilon = \E[|Z|]$.
 
In what follows, we set
$
E_1(\mathcal{D}',m):= \left \lvert f_\text{arr}(\mathcal{D}') - f(\mathcal{D}')\right\rvert
$
to denote the error due to clipping the number of samples per user, and 
\begin{equation}
	E_2^{(\varepsilon)}(m):= \tilde{\Delta}_{f_\text{arr}}/\varepsilon = \frac{ Um}{\varepsilon\cdot \sum_{\ell=1}^L \min\{m_\ell, m\}}
	\label{eq:e2}
\end{equation}
to denote the expected absolute value of noise added to guarantee $\varepsilon$-DP. 
Our intention is to minimize the maximum (or worst-case) error $E^{(\varepsilon)}$ over all datasets. Let $E^{(\varepsilon)}(m):= \max_{\mathcal{D}'} E(\mathcal{D}',m)$ and let
\begin{align}
	\label{eq:minimax1}
	E^{(\varepsilon)}:= \min_{m_\star\leq m\leq m^\star}E^{(\varepsilon)}(m).
\end{align}
We then set
\begin{align}
	\label{eq:uboptimnew}
	{m}^{(\varepsilon)}\in \arg\min_{m_\star\leq m\leq m^\star} E^{(\varepsilon)}(m),
\end{align}
and note that the right-hand side above is a set of cardinality potentially larger than $1$. The next lemma exactly characterizes the worst-case, or maximum, error $E_1(\mathcal{D}',m)$, over all datasets $\mathcal{D}'$. For $\ell\in [L]$, let $\Gamma_\ell:= \min\{m_\ell,m\}$.
\begin{lemma}
	\label{lem:e1}
	We have that 
	\[
	\max_\mathcal{D'} E_1(\mathcal{D}',m)= U\cdot \left(1-\frac{\sum_\ell \Gamma_\ell}{\sum_\ell m_\ell}\right).
	\]
\end{lemma}
\begin{proof}
	First, recall that $m_1\geq m_2\geq \ldots m_L$; let $m_{\ell^\star}$ denote the smallest index $\ell$ such that $m_\ell< m$. Then,
	\begin{align*}
		&E_1(\mathcal{D}',m)\\
		&= \Bigg \lvert \frac{1}{\overline{K}m}\sum_{\ell< \ell^\star}\sum_{j=1}^m S_j^{(\ell)} + \frac{1}{\overline{K}m}\sum_{\ell\geq \ell^\star}\sum_{j=1}^{m_\ell} S_j^{(\ell)} \\
		&\ \ \ \ \ \ \ \ \ \  - \bigg(\frac{1}{\sum_\ell m_\ell} \sum_{\ell<\ell^\star}  \sum_{j=1}^{m_\ell}S_j^{(\ell)}+ \frac{1}{\sum_\ell m_\ell}  \sum_{\ell\geq \ell^\star} \sum_{j=1}^{m_\ell}S_j^{(\ell)}\bigg) \Bigg\rvert\\
		&= \Bigg \lvert \bigg(\frac{1}{\overline{K}m}-\frac{1}{\sum_\ell m_\ell}\bigg)\cdot \sum_{\ell=1}^L \sum_{j=1}^{\Gamma_\ell} S_j^{(\ell)} - \frac{1}{\sum_\ell m_\ell} \sum_{\ell<\ell^\star}\sum_{j>m} S_j^{(\ell)} \Bigg \rvert.
	\end{align*}
	Now, since each of the two terms in the maximization above is non-negative, with $S_j^{(\ell)}\in [0,U]$, for all $\ell\in [L]$ and $j\in [m_\ell]$, we get that 
	\begin{align}
		&\max_{\mathcal{D}'} E_1(\mathcal{D}',m) \notag\\
		&=\max\Bigg\{ \bigg(\frac{U}{\overline{K}m}-\frac{U}{\sum_\ell m_\ell}\bigg) \sum_{\ell=1}^L \Gamma_\ell, \frac{U}{\sum_\ell m_\ell} \sum_{\ell<\ell^\star} (m_\ell-m)^+  \Bigg\}. \label{eq:interminimax}
	\end{align}
	In \eqref{eq:interminimax}, the first expression on the right is attained when $S_j^{(\ell)} = U$, for all $\ell\in [L]$ and $j\in [\Gamma_\ell]$, and $S_j^{(\ell)} = 0$, otherwise. Analogously, the second expression on the right is attained when $S_j^{(\ell)} = U$, for all $\ell<\ell^\star$ and $j>m$, and $S_j^{(\ell)} = 0$, otherwise.
	We use the notation $(c)^+$ to denote $\max\{0,c\}$, for $c\in \mathbb{R}$. Next, observe that $\sum_\ell (m_\ell-m)^+ = \sum_\ell m_\ell - \sum_\ell \Gamma_\ell$. Plugging this into \eqref{eq:interminimax}, we obtain that 
	\begin{align*}
		&\max_{\mathcal{D}'} E_1(\mathcal{D}',m)\\
		&= U\cdot \max\Bigg\{ \bigg(\frac{1}{\overline{K}m}-\frac{1}{\sum_\ell m_\ell}\bigg) \sum_{\ell=1}^L \Gamma_\ell, 1-\frac{\sum_\ell \Gamma_\ell}{\sum_\ell m_\ell}  \Bigg\}\\
		&= U\cdot \left(1-\frac{\sum_\ell \Gamma_\ell}{\sum_\ell m_\ell}\right),
	\end{align*}
	where in the last equality we use the fact that $\overline{K}m = \sum_\ell \Gamma_\ell$.
\end{proof}
Given Lemma \ref{lem:e1}, our objective \eqref{eq:minimax1} yields
\begin{align}
	{E}^{(\varepsilon)}= \min_{m_\star\leq m\leq m^\star} \left( U\cdot \left(1-\frac{\sum_\ell \Gamma_\ell}{\sum_\ell m_\ell}\right)+ \frac{Um}{\varepsilon\cdot \sum_{\ell=1}^L \Gamma_\ell}\right).
	\label{eq:inter2}
\end{align}

Moreover, since 
\begin{align*}
	E^{(\varepsilon)} &= \min_{m_\star\leq m\leq m^\star} \max_{\mathcal{D}'} E(\mathcal{D}',m)\\ &\geq \max_{\mathcal{D}'} \min_{m_\star\leq m\leq m^\star} E(\mathcal{D}',m)
	\geq \min_{m_\star\leq m\leq m^\star} E(\overline{\mathcal{D}},m),
\end{align*}
for \emph{any} dataset $\overline{D}$, we obtain that the total estimation error of the ``best'' pseudo-user based algorithm (via a choice of $m = m_\text{UB}$) for $\overline{D}$, is upper bounded by $E^{(\varepsilon)}$.

Figure \ref{fig:m_e} shows a plot of a minimizer $m^{(\varepsilon)}$ that is the smallest $m\in [m_\star,m^\star]$ that optimizes \eqref{eq:inter2} for the ITMS dataset. We observe that this minimizer appears to be a staircase function (of $\varepsilon$) and is non-decreasing in $\varepsilon$. In what follows, we provide theoretical proofs of these and other properties of the minimizer $m^{(\varepsilon)}$.

\begin{figure}
	\centering
	\includegraphics[width=\linewidth]{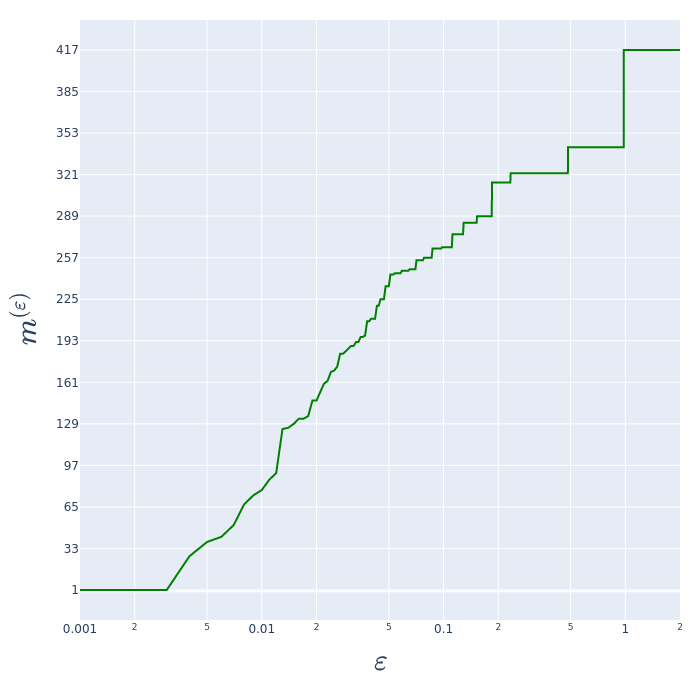}
	\caption{Plot showing the values of $m^{(\varepsilon)}$ obtained by solving \eqref{eq:inter2} for the ITMS dataset with $0.001\le \varepsilon\le 2$. Here, the $\varepsilon$ axis is shown on a log-scale.}
	\label{fig:m_e}
\end{figure}
\begin{lemma}
	For any $\varepsilon>0$, we have that there exists $m^{(\varepsilon)} \in \{m_1,\ldots,m_L\}$.
\end{lemma}
\begin{proof}
	We claim that for any $\varepsilon>0$, a minimizer of \eqref{eq:inter2} occurs at some $m_k$, for $k\in [L]$. Indeed, observe that for any $m\in [m_t,m_{t+1}]$, for $t\in [L-1]$, we have $\sum_\ell \Gamma_\ell = \sum_{\ell=1}^t m_\ell + (L-t)m$. This hence implies that
	\[
	E_1(m):=\max_\mathcal{D'} E_1(\mathcal{D}',m)= U\cdot \left(1-\frac{\sum_\ell \Gamma_\ell}{\sum_\ell m_\ell}\right)
	\]
	is concave in $m$, for $m\in [m_t,m_{t+1}]$.
	
	Likewise, we have
	\[
	E_2^{(\varepsilon)}(m) = \frac{ Um}{\varepsilon\cdot \sum_{\ell=1}^L \min\{m_\ell, m\}} = \frac{m}{c_1 + c_2m},
	\]
	for some constants $c_1, c_2>0$. It can be verified that 
	\[
	\frac{\d^2 E_2^{(\varepsilon)}}{\d m^2} = \frac{-2c_1c_2(c_1+c_2m)}{(c_1+c_2m)^4}<0,
	\]
	implying that $E_2^{(\varepsilon)}(m)$ is also concave in $m$, for $m\in [m_t,m_{t+1}]$.
	
	Putting everything together, we obtain that $E^{(\varepsilon)} = E_1+E_2^{(\varepsilon)}$ is concave in $m$, for $m\in [m_t,m_{t+1}]$, thereby showing that a minimum of $E^{(\varepsilon)}$ occurs at the boundary of the interval $[m_t,m_{t+1}]$. In other words, for $m\in (m_t,m_{t+1})$, we have that $E^{(\varepsilon)}(m) \geq \min\{E^{(\varepsilon)}(m_t),E^{(\varepsilon)}(m_{t+1})\}$, hence showing what we set out to prove.
\end{proof}
The above lemma shows that a brute-force numerical optimization of \eqref{eq:inter2} can be carried out by simply letting $m$ take values in the set $\{m_1,\ldots,m_L\}$.
\begin{lemma}
	\label{lem:minimax1}
	We have that $m^{(\varepsilon)}$ is non-decreasing in $\varepsilon$.
\end{lemma}
\begin{proof}
	We need to prove that $m^{(\varepsilon_2)}\geq m^{(\varepsilon_1)}$, for $\varepsilon_2>\varepsilon_1$. To this end, it is sufficient to prove that for all $m<m^{(\varepsilon_1)}$, it is true that $E^{(\varepsilon_2)}(m)> E^{(\varepsilon_2)}(m^{(\varepsilon_1)})$. Now, for any $m\in [m_\star,m^\star]$, observe that
	\begin{align}
		&E^{(\varepsilon_2)}(m) \notag\\&= \left( U\cdot \left(1-\frac{\sum_\ell \Gamma_\ell}{\sum_\ell m_\ell}\right)+ \frac{ Um}{\varepsilon_2\cdot \sum_{\ell=1}^L \min\{m_\ell, m\}}\right)\notag\\
		&= E^{(\varepsilon_1)}(m)+\frac{Um}{\varepsilon_2\cdot \sum_{\ell=1}^L \min\{m_\ell, m\}}-\frac{ Um}{\varepsilon_1\cdot \sum_{\ell=1}^L \min\{m_\ell, m\}} \notag\\
		&= E^{(\varepsilon_1)}(m)-\frac{ Um(\varepsilon_2-\varepsilon_1)}{\varepsilon_1\varepsilon_2\cdot \sum_{\ell=1}^L \min\{m_\ell, m\}} \label{eq:minimaxinter2}
	\end{align}
Let $g(m):= \frac{ Um(\varepsilon_2-\varepsilon_1)}{\varepsilon_1\varepsilon_2\cdot \sum_{\ell=1}^L \min\{m_\ell, m\}}$. Rewriting $g(m)$ as
\[
g(m) = \frac{ U(\varepsilon_2-\varepsilon_1)}{\varepsilon_1\varepsilon_2\cdot \frac{\sum_{\ell=1}^L \min\{m_\ell, m\}}{m}},
\]
we see that $g$ is an increasing function of $m$. Thus, from \eqref{eq:minimaxinter2}, we obtain that for $m<m^{(\varepsilon_1)}$,
\begin{align*}
	&E^{(\varepsilon_2)}(m)- E^{(\varepsilon_2)}(m^{(\varepsilon_1)})\\
	&=\left(E^{(\varepsilon_1)}(m)-g(m)\right) - \left(E^{(\varepsilon_1)}(m^{(\varepsilon_1)})-g(m^{(\varepsilon_1)})\right)\\
	&= \left(E^{(\varepsilon_1)}(m)-E^{(\varepsilon_1)}(m^{(\varepsilon_1)})\right)+\left(g(m^{(\varepsilon_1)})-g(m)\right)>0,
\end{align*}
where in the last inequality, we have used the definition of $m^{(\varepsilon_1)}$ and the fact that $g(m)$ is increasing in $m$.
\end{proof}

Furthermore, the following lemma also holds.
\begin{lemma}
	\label{lem:minimax2}
	For $\varepsilon = \frac{m_\star}{L\cdot \sum_\ell m_\ell}$, we have that $m^{(\varepsilon)} = m_\star$.
	
	Similarly, for $\varepsilon = \left(\frac{\sum_\ell m_\ell}{Lm_\star}\right)^2$, we have that $m^{(\varepsilon)} = m^\star$.
\end{lemma}
\begin{proof}
	We shall first prove the first statement of the lemma. The second statement relies on very similar arguments and we shall sketch the key details later. 
	
	Now, for the first statement, it suffices to show that for $\varepsilon = \frac{m_\star}{L\cdot \sum_\ell m_\ell}$, we have that $E^{(\varepsilon)}(m+1)>E^{(\varepsilon)}(m)$, for all $m\in [m_\star,m^\star-1]$. It hence follows that the minimizer of $E^{(\varepsilon)}(m)$ occurs at the least allowed value of $m$, which is $m_\star$.
	
	Recall that $$E^{(\varepsilon)}(m) = \min_{m_\star\leq m\leq m^\star} \left( U\cdot \left(1-\frac{\sum_\ell \Gamma_\ell}{\sum_\ell m_\ell}\right)+ \frac{ Um}{\varepsilon\cdot \sum_{\ell=1}^L \Gamma_\ell}\right).$$ Hence, our task reduces to showing that
	\begin{align*}
	 &\frac{ (m+1)}{\varepsilon\cdot \sum_{\ell=1}^L \min\{m+1,m_\ell\}} - \frac{m}{\varepsilon\cdot \sum_{\ell=1}^L \min\{m,m_\ell\}} \\
	 &\ \ \ \  \ \ \ \ >\frac{\sum_\ell \min\{m+1,m_\ell\}}{\sum_\ell m_\ell} - \frac{\sum_\ell \min\{m,m_\ell\}}{\sum_\ell m_\ell}.
	\end{align*}
	Call the expression on the left-hand side above as $\alpha^{(\varepsilon)}(m)$ and that on the right-hand side above as $\beta(m)$. In what follows, we shall show that $\alpha^{(\varepsilon)}(m)\geq \frac{L}{\sum_\ell m_\ell}>\beta(m)$. To this end, observe that
	\begin{align}
		\beta(m)&= \frac{1}{\sum_\ell m_\ell}\cdot \left(\sum_\ell \min\{m+1,m_\ell\} - \sum_\ell \min\{m,m_\ell\}\right) \notag\\
		& = \frac{|S|}{\sum_\ell m_\ell} < \frac{L}{\sum_\ell m_\ell} \label{eq:beta}
	\end{align}
where $S:= \{\ell: m_\ell>m\}$. Let $\gamma_m:= \sum_\ell \min\{m,m_\ell\}$. Note that from above, $\gamma_{m+1} = \gamma_m+|S|$. Hence, 
\begin{align}
	\alpha^{(\varepsilon)}(m)&=  \frac{\gamma_m(m+1)-m\cdot(\gamma_m+|S|)}{\varepsilon\cdot \gamma_{m+1}\gamma_m} \notag\\
	&= \frac{\sum_{\ell\in S^c} m_\ell}{\varepsilon\cdot \gamma_{m+1} \gamma_m}\geq \frac{m_\star}{\varepsilon\cdot \gamma_{m+1}\gamma_m}, \label{eq:alpha}
\end{align}
where in the second equality above, we have used $S^c:= [L]\setminus S$. Now, since $\gamma_m\leq \sum_\ell m_\ell$, for all $m\in [m_\star,m^\star]$, we obtain that
\begin{align*}
	\alpha^{(\varepsilon)}(m)\geq \frac{m_\star}{\varepsilon\cdot \left(\sum_\ell m_\ell\right)^2}\geq \frac{L}{\sum_\ell m_\ell},
\end{align*}
for $\varepsilon = \frac{m_\star}{L\cdot \sum_\ell m_\ell}$, thereby showing what we set out to prove.

To prove the second statement of the lemma, note that it suffices to show that for $\varepsilon = \left(\frac{\sum_\ell m_\ell}{Lm_\star}\right)^2$, we have that $E^{(\varepsilon)}(m+1)<E(m)$, for all $m\in [m_\star,m^\star-1]$.

To prove this claim, we argue that $\alpha^{(\varepsilon)}(m)\leq \frac{\sum_\ell m_\ell}{\varepsilon(Lm_\star)^2}< \beta(m)$. Crucially, to lower bound $\beta(m)$, we bound $|S|$ (see \eqref{eq:beta}) from below by $1$, and to upper bound $\alpha^{(\varepsilon)}(m)$, we upper bound $\sum_{\ell\in S^c} m_\ell$ (see \eqref{eq:alpha}) by $\sum_\ell m_\ell$, and we lower bound $\gamma_m$ by $Lm_\star$. The proof then follows analogous to the proof of the first statement.
\end{proof}
We present next an immediate corollary of Lemmas \ref{lem:minimax1} and \ref{lem:minimax2}, which intuitively states that for small enough $\varepsilon$, the error due to noise addition ($E_2(m)$) dominates over that due to clipping ($\max_{\mathcal{D}'} E_1^{(\varepsilon)}(m,\mathcal{D}')$), thereby setting the optimal value of $m = m_\text{UB}$ in \eqref{eq:minimax1} to be $m_\star$. Likewise, for large enough $\varepsilon$, the error due to clipping dominates over that due to noise addition, thereby setting the optimal value of $m = m_\text{UB}$ in \eqref{eq:minimax1} to be $m^\star$. Let $\varepsilon_\text{min}:= \frac{m_\star}{L\cdot \sum_\ell m_\ell}$ and let $\varepsilon_\text{max}:= \left(\frac{\sum_\ell m_\ell}{Lm_\star}\right)^2$. Then,

\begin{corollary}
	\label{cor:mepsilon}
	For $\varepsilon \leq \varepsilon_\text{min}$, we have that $m^{(\varepsilon)} = m_\star$; likewise, for $\varepsilon\geq \varepsilon_\text{max}$, we have that $m^{(\varepsilon)} = m^\star$.
\end{corollary}

%% file: opt-array-av.tex
\label{sec:opt-array-av}
The discussion in the previous section immediately gives rise to an algorithm, which we call \textsc{OPT-Array-Averaging}, that chooses $m_\text{UB}$ in the \textsc{Array-Averaging} algorithm, in order to jointly optimize the \emph{worst-case} errors due to clipping and privacy. In particular, \textsc{OPT-Array-Averaging} sets $m_\text{UB}$, for a fixed $\varepsilon$, to be that value $m^{(\varepsilon)}$ in \eqref{eq:uboptimnew}. 


While Corollary \ref{cor:mepsilon} shows that for small enough $\varepsilon$, \textsc{OPT-Array-Averaging} sets $m_\text{UB} = m_\star$, the question of explicitly solving for $m^{(\varepsilon)}$ has not been addressed yet. We show next that the function $E_2^{(\varepsilon)}(m)$ in \eqref{eq:e2} is unfortunately non-convex in $m$, for most $\{m_\ell\}$ values of interest, for a fixed $\varepsilon$.  Equivalently, we shall show that the function
\[
\eta(m):= \frac{m}{\sum_{\ell=1}^L \min\{m_\ell, m\}},
\]
is non-convex in $m$, for general datasets. Furthermore, observe from Lemma \ref{lem:e1} that $\max_{\mathcal{D}'} E_1^{(\varepsilon)}(m,\mathcal{D}')$ is \emph{convex} in $m$, for all integer values of $m$. Thus, we have that $E(m)$ as in \eqref{eq:inter2} is non-convex in $m$ in general, and hence minimizing $E(m)$ over $m$ analytically is hard; one therefore has to resort to numerical methods.

We now prove that $\eta(m)$ is non-convex in $m$, for most datasets. Indeed, consider the setting where there exists some $k\in [L]$ such that $m_k-m_{k+1}>2$ (recall that $m_1\geq\ldots\geq m_L$). In what follows, we show that  $\eta(m)$ is in fact \emph{concave}, for $m\in (m_{k+1}, m_k)$, when the argument $m$ is treated as a real number. In particular, this then implies that $\eta(m)$ is non-convex when $m$ takes an integer value in the range $(m_{k+1}, m_k)$.
\begin{lemma}
	We have that $\eta(m)$ is concave in $m$, for $m\in (m_{k+1}, m_k)$.
\end{lemma}
\begin{proof}
	Observe that for $m\in (m_{k+1}, m_k)$, we have that
	\[
	\eta(m) = \frac{m}{c + km},
	\]
	where $c:=\sum_{\ell=k+1}^L m_\ell$. For this range of $m$ values, hence, $\frac{\d^2 \eta}{\d m^2} = -2ck\cdot (c+km)^{-3}$. Since $c,k>0$, we obtain that $\frac{\d^2 \eta}{\d m^2}<0$, implying the concavity of $\eta$ for the given range of $m$ values.
\end{proof}

Given the potential difficulty of solving for $m^{(\varepsilon)}$ in practice, we next present a simpler choice of $m_\text{UB}$ for \textsc{OPT-Array-Averaging}, which results from solving a simpler optimization problem. This new optimization problem  replaces $E_2^{(\varepsilon)}(m)$ by the ``worst-case inverse gain" ${\tilde{\Delta}_{f_\text{arr}}}/{\Delta_f}$ (see \eqref{eq:sens}), which we call $\overline{E_2}^{(\varepsilon)}(m)$, where
\[
\overline{E_2}^{(\varepsilon)}(m):= \max\left\{ \frac{m}{m^\star}, \frac{\overline{m}}{m^\star}\right\},
\]
where $\overline{m}:= \frac{\sum_\ell m_\ell}{L}$. Recall that the worst-case inverse gain is a proxy for the error incurred due to the privacy requirement in \textsc{Array-Averaging}, in comparison with the error incurred by \textsc{Baseline} (which in turn is independent of $m$).

 In our new optimization problem, we minimize the sum of the worst-case (over datasets) error $E_1/U$ and the worst-case inverse gain, over admissible values $m$. More precisely, our optimization problem is as follows
 \begin{align}
 	&\text{minimize}\ \overline{E}(m):=  1-\frac{\sum_\ell \Gamma_\ell}{\sum_\ell m_\ell}+ \max\left\{ \frac{m}{m^\star}, \frac{\overline{m}}{m^\star}\right\}\notag\\
 	&\text{subject to}\ m_\star\leq m\leq m^\star.
 	\label{eq:inter3}
 \end{align}
It can easily be argued that $\overline{E}(m)$ is convex in $m$ and hence \eqref{eq:inter3} is a convex optimization problem with linear constraints. 
 We first state a simple lemma, which characterizes the stationary points of a subderivative of $\overline{E}$. Let $q := \frac{\sum_\ell m_\ell}{m^\star}$ and recall that $m_1\geq\ldots \geq m_L$. 
 \begin{lemma}
 	\label{lem:derivative}
 	We have that $\frac{\d \overline{E}}{\d m} \bigg \rvert_{m=\tilde{m}} = 0$, iff the following conditions hold:
 	\begin{enumerate}
 		\item $q$ is an integer,
 		\item $m_q\geq \overline{m}$, and
 		\item $\tilde{m} = m_q$.
 	\end{enumerate}
 \end{lemma}
\begin{proof}
	First, we take the subderivative of $\overline{E}$ (see also \cite{amin}), to obtain
	\begin{align*}
		\frac{\d \overline{E}}{\d m}= \frac{-|\{\ell: m_\ell\geq m\}|}{\sum_\ell m_\ell}+\frac{1}{m^\star}\cdot \mathds{1}\left\{m\geq \overline{m}\right\}.
	\end{align*}
	Hence, we have that $\frac{\d \overline{E}}{\d m} \bigg \rvert_{m=\tilde{m}} = 0$ only if
	$
	|\{\ell: m_\ell\geq m\}| = \frac{\sum_\ell m_\ell}{m^\star};
	$
	in other words, $\tilde{m}$ is a stationary point if and only if the conditions stated in the lemma hold.
\end{proof}
There is hence a simple procedure, based on the necessity of the KKT conditions (see \cite[Sec. 5.5.3]{boyd}), for obtaining the minimizer in \eqref{eq:inter3}, which we call $\overline{m}^{(\varepsilon)}$.
\begin{enumerate}
	\item[(i)] If $q$ and $m_q$ obey the conditions in Lemma \ref{lem:derivative}, set $\overline{m}^{(\varepsilon)} = m_q$.
	\item[(ii)] Else, we have that either $\overline{m}^{(\varepsilon)} = m_\star$ or $\overline{m}^{(\varepsilon)}= m^\star$ (i.e., $\overline{m}^{(\varepsilon)}$ is a boundary point). Pick $\overline{m}^{(\varepsilon)} \in \arg \min\limits_{m\in \{m_\star,m^\star\}} \overline{E}(m)$.
\end{enumerate}
\begin{remark}
	We have that $\overline{E}(m^\star) = 1$ and 
	$
	\overline{E}(m_\star) = \left(1-\frac{\sum_\ell \Gamma_\ell}{\sum_\ell m_\ell}\right)+ \frac{\overline{m}}{m^\star}.
	$
	Hence, without additional information on the distribution of $\{m_\ell\}_{\ell\geq 1}$, it is not possible to comment on which of $\overline{E}(m_\star)$ or $\overline{E}(m^\star)$ is smaller, in Step (ii) above.
\end{remark}

%% file: conclusion.tex
\label{sec:conclusion}
In this paper, we proposed algorithms  for the private release of sample means of real-world datasets, with particular focus on traffic datasets that contain bus speed samples. {The speed samples in the datasets in consideration are potentially non-i.i.d. with an unknown distribution, and the number of speed samples contributed by different buses are, in general, different.} We analyzed the performance of the different algorithms proposed, via extensive experiments on real-world ITMS datasets and on large synthetic datasets. We then provided theoretical justification for the choices of subroutines used, and recommended subroutines to be used for large datasets. 

Finally, we presented a ``minimax'' analysis of the total estimation error due to clipping and due to noise addition for privacy, in the general setting of pseudo-user creation-based algorithms that clip the number of samples per user, and discussed some interesting consequences. In particular, we obtained an upper bound on the total error incurred by the ``best'' pseudo-user creation-based algorithm on \emph{any} dataset. We then presented a novel procedure, based on the creation of pseudo-users, which clips the number of samples contributed by a user in such a manner as to optimize the total \emph{worst-case error}. 

{These algorithms are readily applicable to general spatio-temporal datasets for releasing a differentially private mean of a desired value.} An interesting line of future research would be to extend the results in this work to the private release of sample means from several distinct location grids simultaneously, in such a manner as to reduce the privacy loss due to the composition of several mechanisms (see, e.g., \cite{kairouz,steinkechapter}).